\documentclass[dvips,showkeys]{article}
\usepackage{engrec}
\usepackage{amsmath,amssymb,amsthm}
\usepackage{authblk}

\newtheorem{theorem}{Theorem}[section]
\newtheorem{lemma}{Lemma}[section]
\newtheorem{assumption}{Assumption}[section]
\theoremstyle{definition}
\newtheorem{definition}{Definition}[section]

\numberwithin{equation}{section}

\usepackage{graphicx}

\makeatletter
\def\@biblabel#1{}

\begin{document}

\title{\normalsize \bf THE THEORETICAL PRICE OF A SHARE-BASED PAYMENT WITH PERFORMANCE CONDITIONS AND IMPLICATIONS FOR THE CURRENT ACCOUNTING STANDARDS}

\author{\footnotesize MASAHIRO FUJIMOTO}
\affil{\footnotesize \textit{Fujimoto Financial Quantitative Research}\\
\textit{fujimoto.ffqr@gmail.com}}
\date{\footnotesize \today}

\maketitle

\begin{abstract}
\footnotesize Although the growth of share-based payments with performance conditions (hereafter, \textit{SPPC}) is prominent today, the theoretical price of SPPC has not been sufficiently studied. Reflecting such a situation, the current accounting standards for share-based payments issued in 2004 have had many problems. This paper develops a theoretical SPPC price model with a framework for a marginal utility-based price, which previous studies proposed is the price of contingent claims in an incomplete market. This paper's contribution is fivefold. First, we restricted the stochastic process to a certain class to demonstrate how to consistently change all variables' probability distributions, which affect the SPPC payoff. Second, we explicitly indicated not only the stochastic processes of the stock price process and performance variables under the changed probability, but also how the changes in the performance variables' drift coefficients related to stock betas. Third, we proposed a convenient model in application that uses only a few parameters. Fourth, we provided a method to estimate the parameters and improve the estimation of both the price and parameters. Fifth, we illustrated the problems in current accounting standards and indicated how the theoretical price model can significantly improve them.

\end{abstract}

\section{Introduction }

Share-based payments with performance conditions (hereafter, \textit{SPPC}) have experienced prominent growth. According to Pay Governance LLC,\footnote{http://paygovernance.com/considering-performance-stock-options}  “Today, more than 80\% of S\&P 500 companies use a variety of LTI performance plans.” The SPPC provides some benefits, such as stocks or stock options, which are vested when such performance variables as net profit or the earnings per share achieve predetermined goals.

However, unlike the standard stock option, the SPPC's theoretical price has not been sufficiently studied. In fact, a May 2018 Google Scholar search for a combination of “performance conditions” and “theoretical price” as well as “performance conditions” and “share-based payment” and “theoretical value” produced virtually no results.

Reflecting such a theoretically unresolved state, the current accounting standards on share-based payments, issued in 2004, are as follows: The fair value under the current standards is measured without performance conditions on a grant date, and whether the fair value is recognized as compensation cost is left to whether the company believes achieving the goal is probable. When judged as probable, the fair value is allocated as compensation cost over the relevant periods, and if the goal is ultimately impossible, the already recognized costs are immediately reversed. Alternatively, when this is judged as improbable, no expenses are recognized, and when the goal is ultimately possible, the unrecognized costs are immediately recognized.

Current accounting standards have various problems due to the uncertainty of cost recognition, and as the recognition at the grant date is left to the company's judgement. Such issues include volatile compensation cost, over-recognition of compensation cost, inconsistency with accounting objectives, a distortion of the company's optimal selection of an award of equity instruments and the significant volatility of a difficult project's compensation cost. These problems can substantially improve if we use this paper's theoretical pricing as the fair value and always recognize compensation cost on the grant date; we will discuss this in detail in Section 5.

The SPPC's payoff depends on performance variables as well as stock prices. As performance variables cannot be traded on the market, the SPPC's theoretical price is the theoretical price of a contingent claim in an incomplete market. This paper adopts a marginal utility-based price, which previous studies posit is the theoretical price of a contingent claim in an incomplete market (Davis 1997, Hugonnier \textit{et al.} 2005). It is a price at which an investor—who maximizes his or her expected utility by only trading with a money market account and tradable stocks—cannot improve the expected utility by buying or selling the contingent claim.

Previous studies' primary results that relate to the theme of this paper are as follows:
\begin{arabiclist}
\item The marginal utility-based price is a general concept for both complete and incomplete markets, as it coincides with arbitrage-free pricing in a complete market.
\item The marginal utility-based price is the expected value of the contingent claim payoff's present value multiplied by a random variable. This random variable is an optimal solution to the dual problem associated with the expected utility maximization (Hugonnier \textit{et al.} 2005).
\item Generally, the optimal solution to the dual problem depends on the utility function and initial wealth 
\end{arabiclist}

However, prior works have not sufficiently and explicitly studied how to solve the dual problem and obtain a marginal utility-based price. This paper's contribution is fivefold.

\begin{arabiclist}
\item We restrict the stock price and performance variable's stochastic processes to a class driven by Brownian motion. With an additional assumption that is unrestrictive in its application, we explicitly solve the dual problem and reveal that the optimal solution depends neither on the utility function nor the initial wealth, and that the optimal solution is a Radon-Nikodým derivative of the new probability $\mathbb{Q}$ with respect to the original probability $\mathbb{P}$.
\item We then explicitly illustrate the stochastic stock price and performance variable processes under $\mathbb{Q}$. The processes' drift coefficients are equal to a quantity obtained by subtracting the product of the stock price's beta (which will be accurately discussed in Section 3) and the expected excess returns of stocks from the original drift coefficients. New drift coefficients of stock price coincide with the risk-free rate minus the dividend yield, which is the same conclusion as in the complete market. The new discovery in this work is that the change in the performance variable's drift coefficient relates to the stock's betas.
\item We demonstrate that with an additional assumption, which is unrestrictive in its application, we can obtain a convenient model that uses only a few parameters.
\item As the SPPC's theoretical price is the expected value under $\mathbb{Q}$, we must use the Monte Carlo method in most cases. Thus, we propose some control variables to improve the estimation accuracy of the theoretical price as well as the parameters used in the model.
\item We demonstrate that existing accounting standards are highly problematic due to the SPPC's lack of a theoretical price; we can greatly improve these standards by adopting this paper's theoretical price as the fair value.
\end{arabiclist}

The paper is organized as follows: Section 2 summarizes the previous studies for both complete and incomplete markets, to the extent necessary to analyze the SPPC's price. Section 3 restricts the stochastic process to the class driven by Brownian motion with some additional assumptions, and we derive primary theorems as the basis to calculate the price. Section 4 formulates the SPPC and derives its theoretical price. Further, we introduce a \textit{period-product}, a useful concept, and explain its use. Section 5 analyzes the problems with current accounting standards and how this paper's theoretical price model could significantly improve them. Section 6 concludes.

\section{Previous Studies' Primary Results  }
First, we explain the concept of a marginal utility-based price (Davis 1997). Suppose an investor maximizes the expected utility from the terminal wealth by investing the initial wealth in a money market account and stocks (or \textit{tradable assets}). Further, suppose the investor receives a proposal to buy or sell a contingent claim at price $p$. If the investor can improve the expected utility by buying some amount of the claims, we consider $p$ as inexpensive. Conversely, if the investor can improve the expected utility by selling some amount of the claim short, we consider $p$ as expensive. A fair price $p$ based on the investor's expected utility is a price at which the investor neither buys nor sells the claim short. If $p$ is at such a level, the marginal utility due to buying (or selling short) a small amount $q$ of the claim is equal to the marginal utility due to the decrease (or increase) of wealth invested in the tradable assets. This is the marginal utility-based price.

In a complete market, the marginal utility-based price and an arbitrage-free price coincide. The payoff of the claim itself does not affect investors' expected utility, as it can be completely replicated by trading the tradable assets. Only the difference between $p$ and the arbitrage-free price (the replication cost) $c$ matters. If $p$ is higher (or lower) than $c$, the investors can use $p-c$ (or $c-p$, respectively) to improve the expected utility. The price at which the investor cannot improve the expected utility coincides with an arbitrage-free price.

Thus, we formulate the following: One money market account and $m$ stocks exist, which are tradable on the market; we call these tradable assets. Their price processes on the filtered probability space $\left( \Omega ,\ \mathcal{F},\ {{\left( {{\mathcal{F}}_{t}} \right)}_{t\in \left[ 0,T \right]}},\ \mathbb{P} \right)$ are the adapted semi-martingale ${{S}_{0}}$ and ${{S}_{i}}\text{ }\left( 1\le i\le m \right)$, respectively. We denote by ${{D}_{i}}\text{ }\left( 1\le i\le m \right)$, an accumulated dividend process that expresses the total dividend from time $0$ to time $t$. We denote all stock price processes and total accumulated dividend processes by the ${{R}^{m}}$-valued processes $\mathbf{S}={{\left( {{S}_{1}},\cdots ,{{S}_{m}} \right)}^{\top }}$ and $\mathbf{D}={{\left( {{D}_{1}},\cdots ,{{D}_{m}} \right)}^{\top }}$, respectively. Further, $\mathcal{F}={{\mathcal{F}}_{T}}$ and ${{\mathcal{F}}_{t}}$ satisfies the usual conditions (right continuous, and ${{\mathcal{F}}_{0}}$ contains all null sets of $\mathcal{F}$).

We introduce the following definitions:
\begin{definition}[Trading Strategy]
The stochastic processes ${{H}_{0}}$, $\mathbf{H}={{\left( {{H}_{1}},\ldots ,{{H}_{m}} \right)}^{\top }}$ are predictable processes, respectively representing the holding amounts of a money-market account and stocks; thus, we call $\left( {{H}_{0}},{{\mathbf{H}}^{\top }} \right)$ or $\mathbf{H}$ a \textit{trading strategy} (or a \textit{strategy}, for brevity).
\end{definition}
\begin{definition}[Admissible Strategy]
A set of wealth processes that can be realized by some strategy with the initial wealth $x$ is

\begin{equation}
\mathcal{X}\left( x \right):=\left\{ X\ge 0;{{X}_{t}}=x+\int_{0}^{t}{{{H}_{0}}_{u}}d{{S}_{0}}_{u}+\int_{0}^{t}{\mathbf{H}_{u}^{\top }}\left( d{{\mathbf{S}}_{u}}+d{{\mathbf{D}}_{u}} \right) \right\}.
\end{equation}

\leftline{We call $X\in \mathcal{X}\left( x \right)$ or $\mathbf{H}$ generating $X$ an \textit{admissible strategy}.}
\end{definition}

When we use ${{S}_{0}}_{t}$ as a numeraire, we denote the relative prices for $1\le i\le m$ by $\tilde{S}{{}_{i}}_{t}:={{{S}_{i}}_{t}}/{{{S}_{0}}_{t}}\;$, the relative accumulated dividend processes by $\tilde{D}{{}_{i}}_{t}:={{{D}_{i}}_{t}}/{{{S}_{0}}_{t}}\;$, and a set of relative wealth processes by $\tilde{\mathcal{X}}\left( x \right):={\mathcal{X}\left( x \right)}/{{{S}_{0}}}\;$. $\mathbf{\tilde{S}}=\left( {{{\tilde{S}}}_{1}},\cdots ,{{{\tilde{S}}}_{m}} \right)$ and $\mathbf{\tilde{D}}=\left( {{{\tilde{D}}}_{1}},\cdots ,{{{\tilde{D}}}_{m}} \right)$. By definition, $\tilde{S}{{}_{0}}_{t}=1$. Further,
\begin{equation}
	\tilde{\mathcal{X}}\left( x \right):=\left\{ \tilde{X}\ge 0;{{{\tilde{X}}}_{t}}=x+\int_{0}^{t}{{{\mathbf{H}}_{u}}\left( d{{{\mathbf{\tilde{S}}}}_{u}}+d{{{\mathbf{\tilde{D}}}}_{u}} \right)} \right\}.	
\end{equation}

{\flushleft{Clearly, $X\in \mathcal{X}\left( x \right)$ and $\tilde{X}\in \tilde{\mathcal{X}}\left( x \right)$ are self-financing portfolios; further, $\mathcal{X}\left( x \right)=x\mathcal{X}\left( 1 \right)$ and $\tilde{\mathcal{X}}\left( x \right)=x\tilde{\mathcal{X}}\left( 1 \right)$.}}

\begin{definition}[Maximal Strategy]
We call $X\in \mathcal{X}\left( x \right)$ a \textit{maximal strategy} if its terminal value cannot be dominated by that of any other strategy in $\mathcal{X}\left( x \right)$, namely, if ${X}'\in \mathcal{X}\left( x \right)$ and ${{X}_{T}}\le {{{X}'}_{T}}$ imply ${X}'=X$.
\end{definition}

\begin{definition}[Acceptable Strategy]
We call a strategy $X$ an \textit{acceptable strategy} if it has a decomposition of the form $X={X}'-{X}''$, where ${X}'$ is an admissible strategy and ${X}''$ is a maximal strategy.
\end{definition}

{\flushleft{For details on maximal and acceptable strategies, see Delbaen \& Schachermayer (1997).}}

We call a probability measure $\mathbb{Q}$ an \textit{equivalent local martingale measure} if it is equivalent to $\mathbb{P}$ and if every $\tilde{X}\in \tilde{\mathcal{X}}\left( 1 \right)$ is a local martingale under $\mathbb{Q}$. We denote by $\mathcal{M}$ the family of all such measures.
\begin{definition}[Set of the Equivalent Local Martingale Measure]
\begin{equation}
	\mathcal{M}:=\left\{ \mathbb{Q}\approx \mathbb{P};\text{every }\tilde{X}\in \tilde{\mathcal{X}}\left( 1 \right)\text{ is a local martingale under }\mathbb{Q} \right\}.	
\end{equation}
\end{definition}

We denote the investor's utility from the terminal wealth ${{X}_{T}}>0$ by the utility function $U:\left( 0,\infty  \right)\to R$. We suppose the investor has some initial wealth $x$ and trades the tradable assets to maximize the expected utility from the terminal wealth. The maximal expected utility for this investor is given by
\begin{equation}
\label{ZEqnNum462184}
	u\left( x \right):=\underset{\tilde{X}\in \tilde{\mathcal{X}}\left( 1 \right)}{\mathop{\sup }}\,\ U\left( x{{S}_{0}}_{T}{{{\tilde{X}}}_{T}} \right).	
\end{equation}

\leftline{We call (2.4) a {\textit{primary problem.}}}

We then define the {\textit{dual problem}} associated with (2.4). The objective function of the dual problem is a conjugate function of $U$ and we denote it by $V:\left( 0,\infty  \right)\to R$:
\begin{equation}
	V\left( y \right):=\underset{x>0}{\mathop{\sup }}\,\ \left( U\left( x \right)-xy \right).	
\end{equation}

\leftline{The constraint set of the dual problem is}

\begin{align}
\tilde{\mathcal{Y}}\left( y \right):=\left\{ \tilde{Y}\ge 0;{{{\tilde{Y}}}_{0}}=y\text{ and }\!\!~\!\!\text{  }\tilde{X}\tilde{Y}={{\left( {{{\tilde{X}}}_{t}}{{{\tilde{Y}}}_{t}} \right)}_{0\le t\le T}}\text{ is a super martingale}
\quad
\right. 
\notag\\
\left. 
\text{ for every }\tilde{X}\in \tilde{\mathcal{X}}\left( 1 \right) \right\}.
\end{align}

\leftline{The dual problem is a minimization problem:}
\begin{equation}
	v\left( y \right):=\underset{\tilde{Y}\in \tilde{\mathcal{Y}}\left( 1 \right)}{\mathop{\inf }}\,\ V\left( y{{{\tilde{Y}}}_{T}}S_{0T}^{-1} \right).	
\end{equation}

We now list the assumptions we will use.
\begin{assumption}
\footnote{Delbaen \& Schachermayer (1994) proved that Assumption 2.1 and “no free lunch with vanishing risk" are equivalent conditions.}
\label{_Ref505931248}
\begin{equation}
	\mathcal{M}\ne \varnothing.
\end{equation}
\end{assumption}

\begin{assumption}
\label{_Ref505931245}
\begin{arabiclist}
\item The utility fuinction $U$ is a strictly increasing, strictly concave, and continuously differentiable function.
\item Further, $U$ satisfies the Inada conditions:
\begin{equation}
\label{ZEqnNum857160}
	\underset{x\to \infty }{\mathop{lim}}\,\ {U}'\left( x \right)=0,\ \underset{x\to 0}{\mathop{lim}}\,\ {U}'\left( x \right)=\infty .	
\end{equation}
\end{arabiclist}
\end{assumption}

\begin{assumption}
\footnote{Kramkov \& Schachermayer (1999) used the investor's utility function of the terminal relative wealth ${{\tilde{X}}_{T}}$. As a result, they formulated the dual problem as ${{\inf }_{\tilde{Y}\in \tilde{\mathcal{Y}}\left( 1 \right)}}E\left[ V\left( y\tilde{Y} \right) \right]$. Further, they use the property ${{\sup }_{h\in D\left( 1 \right)}}E\left[ h \right]\le 1$ (with the variables $h$, $D\left( 1 \right)$, and $C\left( 1 \right)$ as described below as the variables used in their paper) for the proof of Lemmas 3.4 and 3.7 in their paper. This paper used the investor's utility function of the terminal nominal wealth ${{X}_{T}}$. As a result, the dual problem changes to ${{\inf }_{\tilde{Y}\in \tilde{\mathcal{Y}}\left( 1 \right)}}E\left[ V\left( y\tilde{Y}S_{0T}^{-1} \right) \right]$. Assumption 3 is necessary to play the same role as ${{\sup }_{h\in D\left( 1 \right)}}E\left[ h \right]\le 1$ above. In their paper, ${{\sup }_{h\in D\left( 1 \right)}}E\left[ h \right]\le 1$ is not an assumption, but the result of $C\left( 1 \right)$ containing a constant process 1. Alternatively, this paper's Assumption 2.3 cannot be derived from other assumptions, so this needs to be an additional assumption.}

\begin{equation}
\label{_Ref505931245}
	\underset{\tilde{Y}\in \mathcal{Y}\left( 1 \right)}{\mathop{\sup }}\,\ E\left[ {{{\tilde{Y}}}_{T}}S_{0T}^{-1} \right]<\infty . 	
\end{equation}
\end{assumption}

\begin{assumption}
\label{_Ref505931250}
\begin{equation}
	u\left( x \right)<\infty {\text{ for some }}x>0.
\end{equation}
\end{assumption}

The results of previous studies related to this paper's theme involve the following Lemmas 2.1, 2.2, and 2.3.
\begin{lemma}[An Incomplete Market]
\label{_Ref509074698}
Under Assumptions {\upshape 2.1, 2.2, 2.3}, and {\upshape 2.4}, the following claims hold:
\begin{arabiclist}
\item A unique optimal solution ${{Y}^{*}}\left( y \right)\in \mathcal{Y}\left( 1 \right)$ to the dual problem {\upshape(2.7)} exists for any $y>0$.
\item If $\underset{x\to \infty }{\mathop{lim}}\,\ x{{U}'\left( x \right)}/{U\left( x \right)}\;<1$, then a unique optimal solution ${{\tilde{X}}^{*}}\left( x \right)\in \mathcal{X}\left( 1 \right)$ to the primary problem {\upshape(2.4)} exists for any $x>0$. If $y={u}'\left( x \right)$, we have
\begin{equation}
	{U}'\left( {{{x\tilde{X}}}^{*}}\left( x \right){{S}_{0}}_{T} \right){{S}_{0}}_{T}=y{{\tilde{Y}}^{*}}\left( y \right).	
\end{equation}
\end{arabiclist}
\end{lemma}

\begin{proof}
See Kramkov \& Schachermayer's (1999) Theorems 2.1 and 2.2. 		        
\end{proof}
\begin{lemma}[A Complete Market]
\label{_Ref506080083}
Under Assumptions {\upshape 2.1, 2.2, 2.3}, and {\upshape 2.4}, if a market is complete, namely, if $\mathcal{M}$ is a singleton, a unique optimal solution to the primary problem {\upshape(2.1)} exists for any $x>0$. If $y={u}'\left( x \right)$, we have
\begin{equation}
	{U}'\left( {{{x\tilde{X}}}^{*}}\left( x \right) {{S}_{0}}_{T}\right){{S}_{0}}_{T}=y\frac{d\mathbb{Q}}{d\mathbb{P}}	
\end{equation}

\leftline{where ${d\mathbb{Q}}/{d\mathbb{P}}\;$ is a Radon-Nikodým derivative of $\mathbb{Q}$ with respect to $\mathbb{P}$, where $\mathbb{Q}$ is a }
\leftline{unique element of $\mathcal{M}$.}
\end{lemma}

\begin{proof}
See Kramkov \& Schachermayer's (1999) Theorem 2.0.				
\end{proof}

Next, we define a marginal utility-based price. We denote the payoff of a contingent claim that is paid at time $T$ by an ${{\mathcal{F}}_{T}}$ measurable random variable $B$. For $\left( x,q \right)\in {{R}^{2}}$, we denote by $\mathcal{X}\left( x,q\left| B \right. \right)$ the set of acceptable strategies whose initial wealth is $x$ and terminal welth plus $qB$ is non-negative; specifically,

\begin{equation}
	\mathcal{X}\left( x,q\left| B \right. \right):=\left\{ X\text{ is an acceptable strategy with }{{X}_{0}}=x\text{ and }{{X}_{T}}+qB\ge 0 \right\}\text{ }\!\!~\!\!\text{ }\text{.}	
\end{equation}

\begin{definition}[A Marginal Utility-Based Price]
Suppose the claim $B\in {{L}^{0}}$, and $x>0$. The price of the claim $p$ is the marginal utility-based price of $B$ given the initial wealth $x$ if
\begin{equation}
\label{ZEqnNum561875}
	E\left[ U\left( {{X}_{T}}+qB \right) \right]\le u\left( x \right)\text{for}\ \forall q\in R,\forall X\in \mathcal{X}\left( x-pq,q\left| B \right. \right).	
\end{equation}
\end{definition}

The right side is the maximal expected utility from the tradable assets, and the left side is the expected utility when investing $pq$ in the claim with the remaining in the tradable assets. Equation (2.15) signifies that the investor cannot improve the expected utility regardless of the amount $q$ that is added to the portfolio if $p$ is the marginal utility-based price.

The following Lemma \ref{_Ref506067513} combines the expected utility maximization problem with the price of a contingent claim.

\begin{lemma}
\label{_Ref506067513}
Suppose Assumptions {\upshape2.1, 2.2, 2.3} and {\upshape2.4} hold, and $v\left( y \right)<\infty $. Let $X\in \mathcal{X}\left( 1 \right)$ as an arbitrary maximal strategy. If ${{\tilde{Y}}^{*}}\left( y \right)\tilde{X}$ is a uniformly integrable martingale, a contingent claim $B$, such as $\left| B \right|\le a{{X}_{T}}$ for some constant $a>0$, has a unique marginal utility-based price $p\left( B\left| x \right. \right)$. This $p\left( B\left| x \right. \right)$ is given by the following equation:
\begin{equation}
	p\left( B\left| x \right. \right)=E\left[ \tilde{Y}_{T}^{*}\left( y \right)BS_{0T}^{-1} \right].	
\end{equation}
\end{lemma}

\begin{proof}
See Hugonnier \textit{et al.} (2005) Theorem 3.1(i).				 
\end{proof}

If ${{\tilde{Y}}^{*}}\left( y \right)$ is a uniformly integrable martingale, it corresponds to a Radon-Nikodým derivative process of some equivalent local martingale measure $\mathbb{Q}\in \mathcal{M}$; thus, the right side of (2.16) can be expressed as ${{E}^{Q}}\left[ BS_{0T}^{-1} \right]$, where ${{E}^{Q}}\left[ \ \cdot \  \right]$ expresses the expected value under $\mathbb{Q}$. However, ${{\tilde{Y}}^{*}}\left( y \right)$ is generally not necessarily a uniformly integrable martingale, so we cannot express the right side of (2.16) as an expected value under some measure. It is simply the expected value of the payoff's present value multiplied by ${{\tilde{Y}}^{*}}\left( y \right)$ under the original measure.

In a complete market, $\tilde{Y}_{T}^{*}\left( y \right)$ coincides with the unique Radon-Nikodým derivative ${d\mathbb{Q}}/{d\mathbb{P}}\;$. Thus, we observe (2.16) is an extension of the theoretical price formula in a complete market to an incomplete market. 

If $\tilde{Y}_{T}^{*}\left( y \right)$ is specified, we can calculate the theoretical price. However, it is difficult to explicitly calculate $\tilde{Y}_{T}^{*}\left( y \right)$. We solve the dual problem by specifying the set $\tilde{\mathcal{Y}}\left( 1 \right)$ as well as the value function $u\left( \cdot  \right)$ of the primary problem. However, specifying $u\left( \cdot  \right)$ amounts to solving the primary problem, which is only possible after we solve the dual problem. Therefore, the solution is cyclical, and it is not possible to specify $\tilde{Y}_{T}^{*}\left( y \right)$ in a general case.

Previous studies have attempted to solve this problem by restricting the utility function to a certain class. For example, Davis (1997) solved for a log utility function, Frittelli (2000) solved for an exponential utility function, and Henderson (2002) solved for a power and exponential utility function. However, it is seemingly difficult to reach an agreement regarding an application in which these utility functions are appropriate. Further, if we deny these utility functions, we must specify the utility function, its risk aversion parameter, and initial wealth size, which makes it more difficult to reach an agreement a fortiori.

We propose another solution, in that we will restrict the stochastic process instead of the utility function to a certain class. We can indicate that such a restriction is not restrictive, and hence, aggregable in application for the SPPC's theoretical price. Section 3 further explains this solution.

\section{The Stock Price and Performance Variable Model }
\subsection{The stochastic processes}

Henceforth, we suppose the stochastic processes that the $d$-dimensional Brownian motions drive. Consider the setting of (3.1), which consists of one money market account, $m$ stocks, and $d-m$ performance variables.
\begin{equation}
\label{ZEqnNum789836}
	\begin{aligned}
  & d{{S}_{0}}_{t}={{S}_{0}}_{t}{{r}_{t}}dt, \\ 
 & d{{\mathbf{S}}_{t}}=diag\left( {{\mathbf{S}}_{t}} \right)\left( {{\mathbf{b}}_{t}}dt+{{\mathbf{\Sigma }}_{t}}d{{\mathbf{w}}_{t}} \right), \\ 
 & d{{\mathbf{P}}_{t}}=diag\left( {{\mathbf{P}}_{t}} \right)\left( {{\mathbf{c}}_{t}}dt+{{\mathbf{T}}_{t}}d{{\mathbf{w}}_{t}} \right), \\ 
\end{aligned}	
\end{equation}

{\flushleft{where $\mathbf{S},\mathbf{b}$ are $m$-dimensional column vectors; $\mathbf{P},\mathbf{c}$ are $\left( d-m \right)$-dimensional column vectors; $\mathbf{\Sigma }$ is a $m\times d$-dimensional matrix; $\mathbf{T}$ is a $\left( d-m \right)\times d$-dimensional matrix; $\mathbf{w}$ is a $d$-dimensional Brownian motion; and $diag\left( \mathbf{x}  \right)$ is a diagonal matrix with a vector $\mathbf{x}$ as the diagonal elements. The processes $r,\mathbf{b},\mathbf{c},\mathbf{\Sigma }$, and $\mathbf{T}$ are adapted to ${{\mathcal{F}}_{t}}$. For the time variable $t$, we use such notations as ${{\mathbf{b}}_{t}}$ and $\mathbf{b}\left( t \right)$ interchangeably, and sometimes omit $t$ for simplicity.}}

It is worth noting the meanings of performance, as three types of performance variables exist. The first type is a flow variable. For example, consider net profits, defined as the quantity of flow that a firm earns for a certain period. We denote net profit from time $a$ to time $t$ by $N\left( a,t \right)$. The differential of $N\left( a,t \right)$ with respect to $t$ does not depend on the value of $a$, and we denote it by $n\left( t \right):=\partial {N\left( a,t \right)}/{\partial t}\;$. The meaning of $n\left( t \right)$ is an instantaneous rate of net profits at time $t$. We can express net profits for any period $\left[ a,b \right]$ by the equation: 
\begin{equation}
\label{ZEqnNum786518}
	N\left( a,b \right)=\int_{a}^{b}{n\left( t \right)dt}.	
\end{equation}

{\flushleft{Thus, the instantaneous rate of performance variables at time $t$ are useful variables to describe various relationships around the first type performance variables. We adopt this as a basic variable. The performance variables ${{\mathbf{P}}_{t}}$ in (3.1) belonging to this type are such instantaneous rates of performance variables at time $t$. We denote by ${{N}_{1}}$ the set of $i$ such that for $i\in {{N}_{1}}$, ${{P}_{i}}$ belongs to this type.}}

The second performance variable type is one that expresses the state of a certain project, such as the development of a new drug. When modeling this development process, we suppose that several stages exist in judging a success or failure. In order for success at each stage, the performance variable must exceed each threshold; if it fails to exceed the threshold, it cannot proceed. When all thresholds are exceeded, this project eventually succeeds. This process is modeled as follows. There are $n$ time points, such as $0={{t}_{0}}<{{t}_{1}}<{{t}_{2}}<\ldots <{{t}_{n}}\le T$, and the following condition means the $i$th stage is successful:
\begin{equation}
\label{ZEqnNum622135}
	\frac{P\left( {{t}_{i}} \right)}{P\left( {{t}_{i-1}} \right)}\ge {{K}_{i}}.
\end{equation}

{\flushleft{If the condition is not satisfied for some $i$, this project failed at that stage. The ratio at the two time points is used so the state variables can indicate the success or failure at each stage independent from each other. If we formulate (3.3) as $P\left( {{t}_{i}} \right)\ge {{K}_{i}}$, a large value of $P\left( {{t}_{i-1}} \right)$ signifies a high probability of next-stage success, and each step is not independent. We use ${{N}_{2}}$ to denote the set of $i$ such that for $i\in {{N}_{2}}$, ${{P}_{i}}$ belongs to this type.}}

The third performance variable type expresses the instantaneous state at each time, such as the market share. We use ${{N}_{3}}$ to denote the set of $i$ such that for $i\in {{N}_{3}}$, ${{P}_{i}}$ belongs to this type.

When we formulate the SPPC in more detail in Section 4, we will return to the distinction between these three types of performance variables.

The total volatility matrix ${{\left[ \begin{matrix}
   {{\mathbf{\Sigma }}^{\top }} & {{\mathbf{T}}^{\top }}  \\
\end{matrix} \right]}^{\top }}$ is a $d$-dimensional square matrix. We can assume it is a regular matrix without a loss of generality. We use ${{\mathbf{w}}_{1}}$ to denote the first $m$ Brownian motions of $\mathbf{w}$ , ${{\mathbf{w}}_{2}}$ the remaining $d-m$ Brownian motions, and ${{\mathcal{F}}^{\left( 1 \right)}}$ a filtration generated by ${{\mathbf{w}}_{1}}$.
\subsection{An explicit solution to the dual problem}
\begin{assumption}
\label{_Ref505971279}
$r,\mathbf{b}$ and $\mathbf{\Sigma }$ are adapted to ${{\mathcal{F}}^{\left( 1 \right)}}$.
\end{assumption}

Assumption \ref{_Ref505971279} implies $\mathbf{\Sigma }=\left[ \begin{matrix}
   {{\mathbf{\Sigma }}_{1}} & \mathbf{0}  \\
\end{matrix} \right]$. This form of matrix appears to be restrictive, but this is not the case. As the covariance matrix $\mathbf{V}:=\mathbf{\Sigma }{{\mathbf{\Sigma }}^{\top }}$ of the stock prices is a positive definite, an $m$-dimensional square root matrix ${{\mathbf{V}}^{{1}/{2}\;}}$ exists. As ${{\mathbf{V}}^{{1}/{2}\;}}$ can reproduce $\mathbf{V}$, we can use $\left[ \begin{matrix}
   {{\mathbf{V}}^{{1}/{2}\;}} & 0  \\
\end{matrix} \right]$ as the volatility matrix.

Another implication of Assumption 3.1 is to exclude the possibility that the performance variable $\mathbf{P}$ affects $r,\mathbf{b}$ and ${{\mathbf{\Sigma }}_{1}}$. However, such a model has no analytical advantage when evaluating the SPPC's price.

The first result is Theorem \ref{_Ref506079705}.
\begin{theorem}
\label{_Ref506079705}
We define an ${{R}^{m}}$-valued stochastic process \text{\boldmath$\theta$} and an ${{R}^{1}}$-valued stochastic process $Z$ by the following equation

\begin{equation}
\begin{aligned}
 & \text{\boldmath$\theta$}:=\mathbf{\Sigma }_{1}^{-1}\left( \mathbf{b}+\mathbf{d}-r{{\mathbf{1}}_{m}} \right), \\ 
 & {{Z}_{t}}:=\exp \left( -\int_{0}^{t}{{\text{\boldmath$\theta$}^{\top }}}d{{\mathbf{w}}_{1}}_{s}-{\int_{0}^{t}{\left\| \text{\boldmath$\theta$} \right\|}^{2}}ds \right), \\ 
\end{aligned}	
\end{equation}

{\flushleft{where $\mathbf{d}$ is an adapted ${{R}^{m}}$-valued dividend yield process. If Assumptions {\upshape2.1, 2.2, 2.3, 2.4}, and {\upshape3.1} hold, then in setting {\upshape(3.1)} we have}}
\begin{equation}
\label{ZEqnNum438632}
	\tilde{Y}_{T}^{*}\left( y \right)={{Z}_{T}}.	
\end{equation}
\end{theorem}

{\flushleft{We denote by $\mathbb{Q}$ a measure whose Radon-Nikodým derivative ${d\mathbb{Q}}/{d\mathbb{P}}\;$ is ${{Z}_{T}}$. We call $\mathbb{Q}$ an optimal measure.}}

\begin{proof}
As the tradable assets' price processes are adapted to ${{\mathcal{F}}^{\left( 1 \right)}}$ by Assumption \ref{_Ref505971279}, the information ${{\mathcal{F}}_{t}}/\mathcal{F}_{t}^{\left( 1 \right)}$ does not improve the expected utility. We can restrict the admissible wealth processes of the primary problem to those adapted to ${{\mathcal{F}}^{\left( 1 \right)}}$. Therefore, it is possible to solve the primary problem (2.4) with only setting 
\begin{equation}
\label{ZEqnNum405004}
\begin{aligned}
  & d{{S}_{0}}_{t}={{S}_{0}}_{t}{{r}_{t}}dt, \\ 
 & d{{\mathbf{S}}_{t}}=diag\left( {{\mathbf{S}}_{t}} \right)\left( \mathbf{b}dt+{{\mathbf{\Sigma }}_{1}}d{{\mathbf{w}}_{1}}_{t} \right). \\ 
\end{aligned}
\end{equation}

As the maximal values of each primary problem with (3.1) and (3.6) are equal, the uniqueness of the optimal solution to each primary problem—which are established by Lemmas 2.1 and 2.2—leads to an equality between two optimal solutions.

Further, $\mathbb{Q}$ is used to denote the unique equivalent local martingale measure when we consider the primary problem within setting (3.6). The Radon-Nikodým derivative ${d\mathbb{Q}}/{d\mathbb{P}}\;$ coincides with ${{Z}_{T}}$ under setting (3.6); see Theorems 3.6.3 and 3.6.11 (6.23) in the work of Karatzas \& Shreve (1998). Lemma \ref{_Ref506080083} leads to
\begin{equation}
\label{ZEqnNum558848}
	{U}'\left( x\tilde{X}_{T}^{*}\left( x \right){{S}_{0}}_{T} \right){{S}_{0}}_{T}=y{{Z}_{T}}.	
\end{equation}

{\flushleft{Alternatively, by considering a general setting (3.1) in which we consider the performance variables, Lemma \ref{_Ref509074698} leads to}}
\begin{equation}
\label{ZEqnNum636747}
	{U}'\left( x\tilde{X}_{T}^{*}\left( x \right) {{S}_{0}}_{T}\right){{S}_{0}}_{T}=y\tilde{Y}_{T}^{*}\left( y \right).	
\end{equation}

\leftline{We compare the right sides of (3.7) and (3.8) to obtain $\tilde{Y}_{T}^{*}\left( y \right)={{Z}_{T}}$.}
\end{proof}

Next, we calculate the stochastic processes of setting (3.1) under $\mathbb{Q}$. From Gilzanov's theorem, $\mathbf{\hat{w}}$ defined by (3.9) becomes a Brownian motion under $\mathbb{Q}$ (Karatzas \& Shreve 1998 Remark 1.5.3; Karatzas \& Shreve 2012 Section 3.5).
\begin{equation}
	\begin{aligned}
  & d\mathbf{\hat{w}}{{}_{1}}_{t}:=\text{\boldmath$\theta$}dt+d{{\mathbf{w}}_{1}}_{t}, \\ 
 & d\mathbf{\hat{w}}{{}_{2}}_{t}:=d{{\mathbf{w}}_{2}}_{t}, \\ 
\end{aligned}	
\end{equation}

{\flushleft{
where $\left[ \begin{matrix}
   {{\mathbf{T}}_{1}} & {{\mathbf{T}}_{2}}  \\
\end{matrix} \right]=\mathbf{T}$ and ${{\mathbf{T}}_{1}}$ is a $\left( d-m \right)\times m$-dimensional matrix.
}}
Substituting (3.4) and (3.9) into (3.1) creates the following equations:
\begin{equation}
	\begin{aligned}
  & d{{S}_{0}}_{t}={{S}_{0}}_{t}{{r}_{t}}dt, \\ 
 & d{{\mathbf{S}}_{t}}=diag\left( {{\mathbf{S}}_{t}} \right)\left( \left( r{{\mathbf{1}}_{m}}-\mathbf{d} \right)dt+{{\mathbf{\Sigma }}_{1}}d{{{\mathbf{\hat{w}}}}_{1}} \right), \\ 
 & d{{\mathbf{P}}_{t}}=diag\left( {{\mathbf{P}}_{t}} \right)\left( \left( \mathbf{c}-{{\mathbf{T}}_{1}}\mathbf{\Sigma }_{1}^{-1}\left( \mathbf{b}+\mathbf{d}-r{{\mathbf{1}}_{m}} \right) \right)dt+{{\mathbf{T}}_{1}}d{{{\mathbf{\hat{w}}}}_{1}}+{{\mathbf{T}}_{2}}d{{{\mathbf{\hat{w}}}}_{2}} \right).  
\end{aligned}	
\end{equation}

Ultimately, when we change the measure from $\mathbb{P}$ to $\mathbb{Q}$, we must change the drift terms of all variables. The changes to the stock prices' drift terms are the same as those in a complete market. We subtract the expected excess return $\mathbf{b}+\mathbf{d}-r{{\mathbf{1}}_{m}}$ from $\mathbf{b}$. It is noteworthy in (3.10) that we must also change the drift terms of the performance variables. We must subtract from $\mathbf{c}$ the expected excess returns $\mathbf{b}+\mathbf{d}-r{{\mathbf{1}}_{m}}$ multiplied by ${{\mathbf{T }}_{1}}\mathbf{\Sigma }_{1}^{-1}$.

Further, $\mathbb{Q}$ is a so-called minimal martingale measure (MMM); see Bingham \& Kiesel (2013), Section 7.2.3. Previous studies utilized the MMM to find the “local risk minimization strategy”, but did not analyze the relationship between MMM and an optimal solution to the dual problem (Schweizer 1999). This paper's novelty lies in our evidence that the MMM is the optimal solution to the dual problem under setting (3.1) and the assumptions from Theorem \ref{_Ref506107626}.
\footnote{Karatzas \textit{et al}. (1991) proved the optimality of MMM in the case of a power utility function and the “totally unhedgeable” $r,\mathbf{b}$ and ${{\mathbf{\Sigma }}_{1}}$, which are incredibly restrictive assumptions in application. See example 10.2 in the work of Karatzas \textit{et al}. (1991). Another example is a Hull-White stochastic volatility model, which assumes independence between the state variables and stock price. This assumption is restrictive in estimating the SPPC's theoretical price. See Bingham \& Kiesel (2013, p. 316).}

We then explain the meaning of ${{\mathbf{T }}_{1}}\mathbf{\Sigma }_{1}^{-1}$. Fix arbitrary $1\le j\le d-m$. Denote by $m$-dimensional row vector ${{\text{\boldmath$\beta$}}_{j}}_{S}$ the betas of the multiple-regression model, which explain ${d{{P}_{j}}}/{{{P}_{j}}-{{c}_{j}}dt}\;$ (the random term of the instantaneous change rate of the $j$th performance variable) by ${d{{S}_{i}}}/{{{S}_{i}}}\;-{{b}_{i}}dt\ \left( 1\le i\le m \right)$, the $m$ random terms of the stock prices' instantaneous change rates. Transform ${{\mathbf{T}}_{1}}\mathbf{\Sigma }_{1}^{-1}={{\mathbf{T}}_{1}}\mathbf{\Sigma }_{1}^{\top }{{\left( {{\mathbf{\Sigma }}_{1}}\mathbf{\Sigma }_{1}^{\top } \right)}^{-1}}$. The $j$th row of ${{\mathbf{T}}_{1}}\mathbf{\Sigma }_{1}^{\top }$ is a row vector, the $i$th element of which is a covariance between ${d{{P}_{j}}}/{{{P}_{j}}-{{c}_{j}}dt}\;$ and ${d{{S}_{i}}}/{{{S}_{i}}}\;-{{b}_{i}}dt$. The matrix ${{\left( {{\mathbf{\Sigma }}_{1}}\mathbf{\Sigma }_{1}^{\top } \right)}^{-1}}$ is the inverse of the stock returns' covariance matrix. Therefore, the $j$th row of ${{\mathbf{T}}_{1}}\mathbf{\Sigma }_{1}^{-1}$ is ${{\text{\boldmath$\beta$}}_{j}}_{S}$. We call ${{\text{\boldmath$\beta$}}_{j}}_{S}$ multiple-regression betas and denote ${{\mathbf{T}}_{1}}\mathbf{\Sigma }_{1}^{-1}$ by ${{\mathbf{B}}_{P}}_{S}$.

We can observe the product of the multiple regression betas and the expected excess return is then subtracted from the drift term of the performance variable. We summarize the above results as Theorem \ref{_Ref506107626}.

\begin{theorem}
\label{_Ref506107626}
Under Assumptions {\upshape2.1, 2.2, 2.3, 2.4}, and {\upshape3.1}, the setting {\upshape(3.1)} is expressed as the following equations driven by the Brownian motions $\mathbf{\hat{w}}$ under $\mathbb{Q}$:
\begin{equation}
	\begin{aligned}
  & d{{S}_{0}}_{t}={{S}_{0}}_{t}{{r}_{t}}dt, \\ 
 & d{{\mathbf{S}}_{t}}=diag\left( {{\mathbf{S}}_{t}} \right)\left( \left( r{{\mathbf{1}}_{m}}-\mathbf{d} \right)dt+{{\mathbf{\Sigma }}_{1}}d\mathbf{\hat{w}}{{}_{1}}_{t} \right), \\ 
 & d{{\mathbf{P}}_{t}}=diag\left( {{\mathbf{P}}_{t}} \right)\left( \left( \mathbf{c}-{{\mathbf{B}}_{P}}_{S}\left( \mathbf{b}+\mathbf{d}-r{{\mathbf{1}}_{m}} \right) \right)dt+{{\mathbf{T}}_{1}}d\mathbf{\hat{w}}{{}_{1}}_{t}+{{\mathbf{T}}_{2}}d\mathbf{\hat{w}}{{}_{2}}_{t} \right).  
\end{aligned}	
\end{equation}
\end{theorem}

Setting (3.11) indicates that we must subtract the expected returns proportionate to the multiple-regression betas from the original drifts of all variables.
\footnote{For $i$th stock, its multiple-regression betas are defined in the same way as for performance variables. It equals to $m$-dimensional raw vector ${{\mathbf{1}}_{i}}$, the $i$th element of which is 1 and the other element is 0.}

We can use Theorem \ref{_Ref506107626} to obtain the marginal utility based-price of $B$, which satisfies the requirements of Lemma \ref{_Ref506067513} as the expected value of $BS_{0T}^{-1}$ with the stochastic process of (3.11).

\subsection{An optimal portfolio model}

The problem of Theorem \ref{_Ref506107626} in application is that we require $m$ expected excess returns and $\left( d-m \right)\times m$ betas, which are not easily estimated. We can establish a theorem with fewer parameters than Theorem \ref{_Ref506107626} by adding assumptions that are not restrictive in applying the theoretical price of SPPC.
\begin{assumption}
\label{_Ref505935801}
$r,\mathbf{b}$ and ${{\mathbf{\Sigma }}_{1}}$ are deterministic continuous functions of time for $\left[ 0,T \right]$.
\end{assumption}

Assumption \ref{_Ref505935801} excludes the models in which $r,\mathbf{b}$ and ${{\mathbf{\Sigma }}_{1}}$ depend on stock prices. However, such a model has no analytical advantage when evaluating the SPPC's price. Therefore, Assumption \ref{_Ref505935801} is not restrictive, and clearly implies Assumption \ref{_Ref505971279}.

Assumption \ref{_Ref505935801} leads to a strong claim regarding the optimal wealth process. We need additional technical conditions, which we explain in Appendix A, called \textit{Karatzas-Shreve conditions} (Karatzas \& Shreve 1998, Assumptions 3.8.1, 3.8.2).

\begin{lemma}

If Assumptions {\upshape2.1, 2.2, 2.3, 2.4, 3.2} and Karatzas-Shreve conditions hold, the ${{R}^{m}}$-valued process ${{\text{\boldmath$\pi $}}^{*}}$—representing the proportion of the optimal wealth process invested in each stock—can be expressed by
\begin{equation}
\label{ZEqnNum308252}
	{{\text{\boldmath$\pi $}}^{*}}=k{{\left( {{\mathbf{\Sigma }}_{1}}\mathbf{\Sigma }_{1}^{\top } \right)}^{-1}}\left( \mathbf{b}+\mathbf{d}-r{{\mathbf{1}}_{m}} \right)	
\end{equation}

{\leftline{where $k$ is some $R$-valued ${{\mathcal{F}}^{\left( 1 \right)}}$ adapted process.}}
\end{lemma}

\begin{proof}
Karatzas \& Shreve (1998) Theorem 3.8.8 (3.8.24).			        
\end{proof}

We have the following theorem:

\begin{theorem}
Let $b_{\pi }^{*}-r:={{\text{\boldmath$\pi $}}^{*}}^{\top }\left( \mathbf{b}+\mathbf{d}-r{{\mathbf{1}}_{m}} \right)$ be the expected excess return of the optimal wealth process, and $\text{\boldmath$\beta $}{{_{P}^{*}}_{\pi }}:={{\left( \beta {{_{1\pi}^{*}}},\cdots ,\beta {{_{d-m\pi}^{*}}} \right)}^{\top }}$be a $\left( d-m \right)$-dimensional column vector whose $j$th element is a beta of a single-regression model that explains ${d{{P}_{j}}_{t}}/{{{P}_{j}}_{t}}\;-{{c}_{j}}dt\ \left( j=1,\cdots ,d-m \right)$ by ${dX_{t}^{*}}/{X_{t}^{*}}\;-b_{\pi }^{*}dt$, where ${{X}^{*}}$ is the optimal wealth process. If Assumptions {\upshape2.1, 2.2, 2.3, 2.4, 3.2} and Karatzas-Shreve conditions hold, then we have
\begin{equation}
\label{ZEqnNum183557}
	\begin{aligned}
  & d{{S}_{0}}_{t}={{S}_{0}}_{t}{{r}_{t}}dt, \\ 
 & d{{\mathbf{S}}_{t}}=diag\left( {{\mathbf{S}}_{t}} \right)\left( \left( r{{\mathbf{1}}_{m}}-\mathbf{d} \right)dt+{{\mathbf{\Sigma }}_{1}}d\mathbf{\hat{w}}{{}_{1}}_{t} \right), \\ 
 & d{{\mathbf{P}}_{t}}=diag\left( {{\mathbf{P}}_{t}} \right)\left( \left( \mathbf{c}-\text{\boldmath$\beta $}{{_{P}^{*}}_{\pi }}\left( b_{\pi }^{*}-r \right) \right)dt+{{\mathbf{T}}_{1}}d\mathbf{\hat{w}}{{}_{1}}_{t}+{{\mathbf{T}}_{2}}d\mathbf{\hat{w}}{{}_{2}}_{t} \right).  
\end{aligned}	
\end{equation}
\end{theorem}

\begin{proof}
Multiplying ${{\mathbf{\Sigma }}_{1}}\mathbf{\Sigma }_{1}^{\top }$ on both sides of (3.12) from the left, we have
\begin{equation}
\label{ZEqnNum823454}
	{{\mathbf{\Sigma }}_{1}}\mathbf{\Sigma }_{1}^{\top }{{\text{\boldmath$\pi $}}^{*}}=k\left( \mathbf{b}+\mathbf{d}-r{{\mathbf{1}}_{m}} \right).	
\end{equation}

{\leftline{Multiplying ${{\text{\boldmath$\pi $}}^{*}}^{\top }$ on both sides of (3.14) from the left, we have}}
\begin{equation}
\label{ZEqnNum440614}
	{{\text{\boldmath$\pi $}}^{*}}^{\top }{{\mathbf{\Sigma }}_{1}}\mathbf{\Sigma }_{1}^{\top }{{\text{\boldmath$\pi $}}^{*}}=k{{\text{\boldmath$\pi $}}^{*}}^{\top }\left( \mathbf{b}+\mathbf{d}-r{{\mathbf{1}}_{m}} \right)=k\left( b_{\pi }^{*}-r \right).	
\end{equation}

{\leftline{Deleting $k$ from (3.14) and (3.15), we have}}
\begin{equation}
\label{ZEqnNum369679}
	\mathbf{b}+\mathbf{d}-r{{\mathbf{1}}_{m}}=\frac{{{\mathbf{\Sigma }}_{1}}\mathbf{\Sigma }_{1}^{\top }{{\text{\boldmath$\pi $}}^{*}}}{{{\text{\boldmath$\pi $}}^{*}}^{\top }{{\mathbf{\Sigma }}_{1}}\mathbf{\Sigma }_{1}^{\top }{{\text{\boldmath$\pi $}}^{*}}}\left( b_{\pi }^{*}-r \right).	
\end{equation}

{\leftline{Multiplying ${{\mathbf{T }}_{1}}\mathbf{\Sigma }_{1}^{-1}$ on both sides of (3.16) from the left, we have}}

\begin{equation}
	{{\mathbf{T }}_{1}}\mathbf{\Sigma }_{1}^{-1}\left( \mathbf{b}+\mathbf{d}-r{{\mathbf{1}}_{m}} \right)=\frac{{{\mathbf{T}}_{1}}\mathbf{\Sigma }_{1}^{\top }{{\text{\boldmath$\pi $}}^{*}}}{{{\text{\boldmath$\pi $}}^{*}}^{\top }{{\mathbf{\Sigma }}_{1}}\mathbf{\Sigma }_{1}^{\top }{{\text{\boldmath$\pi $}}^{*}}}\left( b_{\pi }^{*}-r \right).	
\end{equation}

{\flushleft{As the $\mathbf{\Sigma }_{1}^{\top }{{\text{\boldmath$\pi $}}^{*}}$ appearing on the right side is a transposition of the volatility vector of the optimal portfolio ${{\text{\boldmath$\pi $}}^{*}}$, ${{\mathbf{T}}_{1}}\mathbf{\Sigma }_{1}^{\top }{{\text{\boldmath$\pi $}}^{*}}$ is a $\left( d-m \right)$-dimensional column vector, the $j$th element of which is a covariance between ${d{{P}_{j}}_{t}}/{{{P}_{j}}_{t}}\;-{{c}_{j}}dt\ \left( j=1,\cdots ,d-m \right)$ and ${dX_{t}^{*}}/{X_{t}^{*}}\;-b_{\pi }^{*}dt$. The matrix ${{\text{\boldmath$\pi $}}^{*}}^{\top }{{\mathbf{\Sigma }}_{1}}\mathbf{\Sigma }_{1}^{\top }{{\text{\boldmath$\pi $}}^{*}}$ in the denominator on the right side is the variance of ${d\hat{X}}/{{\hat{X}}}\;$. Therefore, we have}}
\begin{equation}
\label{ZEqnNum762796}
	{{\mathbf{T }}_{1}}\mathbf{\Sigma }_{1}^{-1}\left( \mathbf{b}+\mathbf{d}-r{{\mathbf{1}}_{m}} \right)=\text{\boldmath$\beta $}{{_{P}^{*}}_{\pi }}\left( b_{\pi }^{*}-r \right).	
\end{equation}

{\leftline{Substituting (3.18) into (3.10) leads to (3.13).}}
\end{proof}

By Theorem 3.3, we need not estimate the expected excess returns and multiple-regression betas of stocks. Instead, it is sufficient to estimate the expected excess returns and single-regression beta of the optimal portfolio. This is because the expected excess returns and the betas of the stocks invested in an optimal portfolio, consistent with those of the optimal portfolio, automatically reflect on the performance variable drift adjustments.

In practice, enough precedent exists to adopt a stock index as the optimal portfolio. Many studies have confirmed that it is difficult to find active management funds that stably outperform the stock index (Carhart 1997, Fama \& French 2010). When considering an investor whose optimal portfolio is the stock index, we can regard such an investor's estimation as both long-term and qualified in the above sense.

As many financial institutions announce prospective market index returns, it is relatively easy to estimate the market index's expected excess returns.

\subsection{Two features of the optimal measure}
The optimal measure implies two features in the stochastic processes (3.11) and (3.13). The first feature is that the performance variable's drift term should change by the expected excess returns on invested stocks or the optimal portfolio multiplied by their betas. Economically, this change occurs because the relative advantage of the investment in the SPPC, compared to that in the tradable assets, affects the SPPC's theoretical price. For an investor who purchases the SPPC by decreasing wealth investing in tradable assets, the higher the stocks' expected excess returns, the higher the SPPC's expected return must be. Ultimately, the theoretical price must decrease; therefore, the probability of achieving goals must also decrease. In adjusting the performance variables' drift terms along with (3.11) or (3.13), the higher the stocks' expected excess returns, the lower the probability that the performance variables will achieve goals and the price; thus, it is consistent. If the price does not decrease, the investor can increase the expected utility by selling the SPPC short.

Neglecting the betas means we assume the investor is risk neutral, which is empirically illogical and leads to the conclusion that the investor can increase the expected utility if the stocks' excess returns are positive. We must be careful not to unconsciously fall into such an unreasonable assumption.

We must estimate betas carefully. If the performance variable includes sales or profits and we use the market index as the optimal portfolio, the betas are positive, as the market index follows the same direction as the economic trend, and sales and profits are affected by economic trends. 

The second feature is that the performance variables' drift adjustments do not depend on the type of utility function. Consider that we used the utility function to define the marginal utility-based price; it is therefore unexpected that the drift adjustment in (3.11) or (3.13) does not need information on the utility function. 

There are two reasons for this occurrence. First, although $\tilde{Y}_{T}^{*}\left( y \right)$ depends on the utility function in general, it does not under Assumption \ref{_Ref505971279} (see Theorem \ref{_Ref506107626}). Second, we define the marginal utility-based price with $q=0$. The marginal utility with respect to $q$ is
\begin{equation}
\label{ZEqnNum779844}
	\frac{\partial }{\partial q}E\left[ U\left( X_{T}^{*}+qB \right) \right]-\frac{\partial }{\partial q}u\left( x-qp \right)=E\left[ {U}'\left( X_{T}^{*}+qB \right)B \right]-{u}'\left( x-qp \right)p.	
\end{equation}

{\flushleft{The first term is a marginal expected utility due to the change in terminal wealth, and the second term is as such due to the change in initial wealth invested in tradable assets. Evaluating (3.19) at $q=0$ leads to}}
\begin{equation}
\label{ZEqnNum122858}
	E\left[ {U}'\left( X_{T}^{*} \right)B \right]-{u}'\left( x \right)p.	
\end{equation}

{\leftline{Substituting (2.12) and $y={u}'\left( x \right)$ into the above equation, we have}}
\begin{equation}
	E\left[ {U}'\left( X_{T}^{*} \right)B \right]-{u}'\left( x \right)p=E\left[ y{{{\tilde{Y}}}^{*}}\left( y \right)B S_{0T}^{-1}\right]-yp=y\left( E\left[ {{{\tilde{Y}}}^{*}}\left( y \right)B S_{0T}^{-1}\right]-p \right).	
\end{equation}

{\leftline{The marginal expected utility equals zero when $p=E\left[ {{{\tilde{Y}}}^{*}}\left( y \right)B S_{0T}^{-1}\right]$.}}

Simultaneously, we can see that if $q\ne 0$, then
\begin{equation}
\label{ZEqnNum358531}
	\frac{\partial }{\partial q}E\left[ U\left( X_{T}^{*}+qB \right) \right]-\frac{\partial }{\partial q}u\left( x-qp \right)\ne y\left( E\left[ {{{\tilde{Y}}}^{*}}\left( y \right)B S_{0T}^{-1}\right]-p \right).	
\end{equation}

{\flushleft{This inequality means that if we define a marginal utility-based price at a non-zero $q$ (the utility indifference price), then $E\left[ {{{\tilde{Y}}}^{*}}\left( y \right)B S_{0T}^{-1}\right]$ is not a price that makes the marginal expected utility equal zero at a non-zero $q$. We can guess that a utility indifference price will depend on the utility function and initial wealth $x$ even if Assumption \ref{_Ref505971279} holds. However, no agreement seemingly exists regarding the type of utility function and what size of initial wealth is appropriate. This disagreement means the marginal utility-based price at $q=0$ is the only sensible definition of the theoretical price.}}

\section{The Theoretical Price of SPPC }
\subsection{Formula of the theoretical price of SPPC}

We then formulate SPPC's theoretical price by considering only stocks and performance variables that affect the SPPC payoff. The difference between stocks and performance variables is that the former can be traded on the market, but the latter cannot. Thus, 
\begin{equation}
\label{ZEqnNum978656}
	\begin{aligned}
  & d{{S}_{0}}_{t}={{S}_{0}}_{t}{{r}_{t}}dt, \\ 
 & d{{\mathbf{S}}_{t}}=diag\left( {{\mathbf{S}}_{t}} \right)\left( \left( {{r}_{t}}{{\mathbf{1}}_{m}}-{{\mathbf{d}}_{t}} \right)dt+{{\mathbf{\Sigma }}_{1}}_{t}d\mathbf{\hat{w}}{{}_{1}}_{t} \right), \\ 
 & d{{\mathbf{P}}_{t}}=diag\left( {{\mathbf{P}}_{t}} \right)\left( {{\text{\boldmath$\mu $}}_{t}}dt+{{\mathbf{T}}_{t}}d{{{\mathbf{\hat{w}}}}_{t}} \right). \\ 
\end{aligned}	
\end{equation}

{\flushleft{where ${{\text{\boldmath$\mu $}}_{t}}$ is the drift coefficient after the beta adjustments in (3.11) or (3.13). We used the same notations as (3.11) in a different meaning to simplify the description. Further, we consider the issuer's stock as well as all stocks that affect the payoff. For example, if the formula for the exercise price is }}
\begin{equation}
	\text{the issuer's stock price on the grant date }\!\!\times\!\!\text{ }\frac{\text{a market index on the vesting date}}{\text{a market index on the grant date}},	
\end{equation}

{\leftline{we include a stock index in the stocks. If the formula for the exercise price is }}
\begin{equation}
\begin{aligned}
   \text{the issuer's stock price on the }&\text{grant date} \\ 
 & \text{ }\!\!\times\!\!\text{ }\frac{\text{a competitor }\!\!'\!\!\text{ s stock price on the vesting date}}{\text{a competitor }\!\!'\!\!\text{ s stock price on the grant date}}, 
\end{aligned}
\end{equation}

{\leftline{we include the competitor's stock in the stocks. }}

Let $m$ and $n$ be the number of stocks and performance variables in (4.1), respectively, and $d:=m+n$. Let ${{T}_{v}}<T$ be the vesting date. 

For $1\le i\le n$, denote by ${{I}_{i}}$ the variable that assumes a value of one when the condition of the $i$th performance variable is satisfied, and zero when it is not. We can describe ${{I}_{i}}$ by using indicator functions and denoting each goal by ${{K}_{i}}$, as follows:
\begin{equation}
	{{I}_{i}}={{1}_{\left\{ {{P}_{iSUM}}\ge {{K}_{i}} \right\}}}\text{ for }i\in {{N}_{1}}\text{ }	
\end{equation}

{\leftline{where ${{P}_{iSUM}}:=\int_{{{a}_{i}}}^{{{b}_{i}}}{{{P}_{iu}}}du\text{ and }\left[ {{a}_{i}},{{b}_{i}} \right]\subset \left[ 0,{{T}_{v}} \right]\text{ }$is the accounting period,}}
\begin{equation}
	{{I}_{i}}={{1}_{\left\{ {{{P}_{i}}\left( {{t}_{i}}_{k} \right)}/{{{P}_{i}}\left( {{t}_{i}}_{k-1} \right)}\;\ge {{K}_{i}} \right\}}},\ 1\le k\le {{L}_{i}}\text{ for }i\in {{N}_{2}}	
\end{equation}

{\flushleft{where $\left( 0={{t}_{i}}_{_{0}}<\ldots <{{t}_{i}}_{_{1}}<\ldots <{{t}_{i}}_{_{{{L}_{i}}}} \right)\subset \left[ 0,{{T}_{v}} \right]$ are time points to judge success or failure, and}}
\begin{equation}
\label{ZEqnNum171476}
	\text{ }{{I}_{i}}={{1}_{\left\{ {{P}_{i}}\left( {{T}_{v}} \right)\ge {{K}_{i}} \right\}}}\text{ for }i\in {{N}_{3}}.	
\end{equation}

{\leftline{We call ${{P}_{iSUM}}$ a \textit{period-sum.}}}

We use $C$ to denote the payoff determined at time $T$ of the award if the conditions are satisfied. Moreover, $C$ is an $\mathcal{F}_{T}^{\left( 1 \right)}$ measurable random variable. The form of $C$ depends on the award's content. If the award provides stocks, then $C={{S}_{T}}$, where ${{S}_{T}}$ is the issuer's stock price at time $T$. If the award provides stock options, the form of $C$ can be specified only after the option's content is specified. It might be a simple American call option with an exercise price $K$ and payoff $\max \left( {{S}_{t}}-K,0 \right)\exp \left( \int_{t}^{T}{{{r}_{u}}du} \right)$ for ${{T}_{v}}\le t\le T$, or it might be an exotic option (see Shreve 2004 for the pricing of such options).

The SPPC payoff is
\footnote{
If the excise price depends on the performance condition, the payoff of such an option is $\max ({{S}_{T}}-{{K}_{1}}{{1}_{\left\{ P\le K \right\}}}-{{K}_{2}}{{1}_{\left\{ P\ge K \right\}}})$ where ${{K}_{1}}$ and ${{K}_{2}}$ are exercise prices and $K$ is a performance goal, which is the same as $\max ({{S}_{T}}-{{K}_{1}}){{1}_{\left\{ P\le K \right\}}}+\max ({{S}_{T}}-{{K}_{2}}){{1}_{\left\{ P\ge K \right\}}}$. Thus, such an option is a group of options with a payoff of $\max ({{S}_{T}}-{{K}_{1}}){{1}_{\left\{ P\le K \right\}}}$ and $\max ({{S}_{T}}-{{K}_{2}}){{1}_{\left\{ P\ge K \right\}}}$. We can analyze this in the same manner as (4.7).
}

\begin{equation}
\label{ZEqnNum910334}
	B:=C\prod\limits_{i=1}^{n}{{{I}_{i}}} .	
\end{equation}

{\flushleft{The product of ${{I}_{i}}$s means that multiple goals must be simultaneously achieved for the award to be vested.}}

We can observe that (4.7) will satisfy the requirements of Lemma \ref{_Ref506067513}. First, we will demonstrate there is a maximal strategy $X\in \mathcal{X}\left( 1 \right)$, satisfying the condition $B\le a{{X}_{T}}$ with some constant $a>0$.

As $C$ is adapted to $\mathcal{F}_{T}^{\left( 1 \right)}$, using the standard discussion of the complete market reveals there is a wealth process $X\in \mathcal{X}\left( 1 \right)$ such that $C=x{{X}_{T}}$ where $x={{E}^{Q}}\left[ CS_{0T}^{-1} \right]$; thus, $B=C\prod\limits_{i=1}^{n}{{{I}_{i}}}\le C=x{{X}_{T}}$. Further, we indicate such $X$ is a maximal strategy. If this $X$ is not the case, then an admissible strategy ${X}'\in \mathcal{X}\left( 1 \right)$ exists that dominates $X$, namely, ${{X}_{T}}\le {{{X}'}_{T}}$ and ${{X}_{T}}<{{{X}'}_{T}}$ with some positive probability. Multiplying both sides by $S_{0T}^{-1}{{Z}_{T}}$ and taking the expected value of both sides under $\mathbb{P}$, we have ${{E}^{Q}}\left[ {{{\tilde{X}}}_{T}} \right]<{{E}^{Q}}\left[ {{{\tilde{{X}'}}}_{T}} \right]$. However, as both $X$ and ${X}'$ are elements of $\mathcal{X}\left( 1 \right)$, ${{E}^{Q}}\left[ {{{\tilde{X}}}_{T}} \right]={{E}^{Q}}\left[ {{{\tilde{{X}'}}}_{T}} \right]=1$, which is a contradiction. Therefore, $X$ is a maximal strategy.

Next, we will demonstrate that regarding the above maximal strategy $X$, $Z\tilde{X}$ is a uniformly integrable martingale. The relative wealth process $\tilde{X}$ satisfies
\begin{equation}
	d{{\tilde{X}}_{t}}={{\tilde{X}}_{t}}\text{\boldmath$\pi $}{{\mathbf{\Sigma }}_{1}}d\mathbf{\hat{w}}{{}_{1}}_{t}	
\end{equation}

{\flushleft{where $\text{\boldmath$\pi $}$ represents the proportion of current wealth invested in each stock (Shreve 2004, 5.2.27). As the Ito Integral is a martingale (Shreve 2004, Theorem 4.3.1(iv)), $\tilde{X}$ is a $\mathbb{Q}$-martingale and $Z\tilde{X}$ is a $\mathbb{P}$-martingale. Thus, we have ${{Z}_{t}}{{\tilde{X}}_{t}}=E\left[ {{Z}_{T}}{{{\tilde{X}}}_{T}}\left| {{\mathcal{F}}_{t}} \right. \right]$. Further, $E\left[ \left| {{Z}_{T}}{{{\tilde{X}}}_{T}} \right| \right]<\infty $, because}}
\begin{equation}
	E\left[ \left| {{Z}_{T}}{{{\tilde{X}}}_{T}} \right| \right]=E\left[ {{Z}_{T}}{{{\tilde{X}}}_{T}} \right]=E\left[ {{Z}_{T}}{{{\tilde{X}}}_{T}}\left| {{\mathcal{F}}_{0}} \right. \right]={{Z}_{0}}{{\tilde{X}}_{0}}=1.	
\end{equation}

{\flushleft{The conditional expectation of a random variable in ${{L}^{1}}$, given ${{\mathcal{F}}_{t}}$, is a uniformly integrable martingale (Williams 1991, 13.4).}}

As we have confirmed that the SPPC's payoff satisfies all requirements of Lemma \ref{_Ref506067513}, we can apply Lemma \ref{_Ref506067513} to the SPPC's payoff. The SPPC's price is
\begin{equation}
	p:={{E}^{Q}}\left[ S_{0T}^{-1}C\prod\limits_{i=1}^{n}{{{I}_{i}}} \right].	
\end{equation}

\subsection{How to improve the estimation}

As no analytical formula exists to express the probability distribution of the period-sum ${{P}_{iSUM}}$, we must calculate (4.10) using the Monte Carlo method if the SPPC contains a type 1 performance variable.

The concept of a {\textit{period-product}} of the performance variable is useful in evaluating the theoretical price. We define a period-product of the performance variable as 
\begin{equation}
	{{P}_{iPROD}}:=\left( {{b}_{i}}-{{a}_{i}} \right)\exp \left( \frac{1}{{{b}_{i}}-{{a}_{i}}}\int_{{{a}_{i}}}^{{{b}_{i}}}{\log {{P}_{iu}}}du \right).	
\end{equation}

{\flushleft{We can indicate the meaning of ${{{P}_{iPROD}}}/{\left( {{b}_{i}}-{{a}_{i}} \right)}\;$ as a continuous time version of the geometric mean of the performance variables at discrete time points. Divide a period $[a,b]$ into ${{t}_{0}}={{a}_{i}},\cdots ,{{t}_{k}}={{a}_{i}}+k\Delta t,\cdots {{t}_{n}}={{b}_{i}}$ $\left( k=0,\cdots ,n \right)$ where $\Delta t:={\left( {{b}_{i}}-{{a}_{i}} \right)}/{n}\;$. The geometric mean of $\left( {{P}_{i}}{{_{a}}_{_{i}}},\cdots ,{{P}_{i}}_{{{t}_{k}}},\cdots ,{{P}_{i}}{{_{b}}_{_{i}}} \right)$ is}}
\begin{equation}
\label{ZEqnNum122569}
	{{\left( \prod\limits_{k=0}^{n}{{{P}_{{{t}_{k}}}}} \right)}^{\frac{1}{n+1}}}=:G.	
\end{equation}

{\leftline{Then}}
\begin{equation}
\begin{aligned}
  \log G=\frac{1}{n+1}\sum\limits_{k=0}^{n}{\log {{P}_{{{t}_{k}}}}}
 & =\frac{n}{n+1}\frac{1}{{{b}_{i}}-{{a}_{i}}}\sum\limits_{k=0}^{n}{\log {{P}_{{{t}_{k}}}}}\Delta t \\ 
 & \to \frac{1}{{{b}_{i}}-{{a}_{i}}}\int_{{{a}_{i}}}^{{{b}_{i}}}{\log {{P}_{iu}}}du\text{ as }n\to \infty . \\ 
\end{aligned}
\end{equation}

{\leftline{Thus, we have}}
\begin{equation}
	G\to \exp \left( \frac{1}{{{b}_{i}}-{{a}_{i}}}\int_{{{a}_{i}}}^{{{b}_{i}}}{\log {{P}_{iu}}}du \right)=\frac{{{P}_{iPROD}}}{{{b}_{i}}-{{a}_{i}}}\text{ }\!\!~\!\!\text{ as }n\to \infty \text{.}	
\end{equation}

Henceforth, we will call the original SPPC conditional to the period-sum the {\textit{period-sum payment}} and the SPPC conditional to the period-product the {\textit{period-product payment}}. We assume ${{L}_{i}}=1\ for \ i\in {{N}_{2}}$ to simplify the description. Let ${{N}_{0}}:={{\left( i \right)}_{1\le i\le m}}$. The payoff of the period-product payment is
\begin{equation}
\label{ZEqnNum961749}
	C\prod\limits_{i=1}^{n}{{{{{I}'}}_{i}}}	
\end{equation}

{\leftline{where ${{{I}'}_{i}}={{1}_{\left\{ {{P}_{iPROD}}\ge {{K}_{i}} \right\}}}$ for $i\in {{N}_{1}}$ and ${{{I}'}_{i}}={{I}_{i}}$ for $i\in {{N}_{2}}\cup {{N}_{3}}$.}}

As we will note in Appendix B, a $d$-dimensional random vector
\begin{equation}
\label{ZEqnNum488418}
	\mathbf{x}:={{\left( {{\left( \log {{S}_{i}}\left( T \right) \right)}_{i\in {{N}_{0}}}},{{\left( \log {{P}_{iPROD}} \right)}_{i\in {{N}_{1}}}},{{\left( \log {{{P}_{i}}\left( {{t}_{i}}_{1} \right)}/{{{P}_{i}}\left( {{t}_{i}}_{0} \right)}\; \right)}_{i\in {{N}_{2}}}},{{\left( \log {{P}_{i}}\left( {{T}_{v}} \right) \right)}_{i\in {{N}_{3}}}} \right)}^{\top }}	
\end{equation}

{\flushleft{has a $d$-dimensional normal distribution, for which an analytical formula using instantaneous parameters exists.}}

The period-product payment price is
\begin{equation}
	{{p}_{PROD}}:={{E}^{Q}}\left[ S_{0T}^{-1}C\prod\limits_{i=1}^{n}{{{{{I}'}}_{i}}} \right].	
\end{equation}

We can calculate ${{p}_{PROD}}$ more quickly and accurately by using the analytical distribution formula of $\mathbf{x}$ rather than the Monte Carlo method, which uses random numbers step by step with many paths. We call this price estimated with the analytical distribution formula a {\textit{quasi-analytical theoretical price}}. In the case of $n=1$ and $m=1$, we can calculate this by combining the analytical formula and numerical integral. In other cases, we can directly generate the samples of $\mathbf{x}$ as random numbers. This calculation takes much less time to obtain the same number of samples than does the case in which we generate random numbers step by step with paths from $t=0$ to $t=T$.

We can utilize ${{p}_{PROD}}$ as a control variable to improve the accuracy in estimating the period-sum payment theoretical price $p$ according to the following formula:
\begin{equation}
	\begin{aligned}
  & \text{the estimation of the theoretical price of the period-sum payment} \\ 
 & =\text{ }\!\!~\!\!\text{ the average of the theoretical price of the period-sum payment over paths} \\ 
 & +\text{the quasi-analytical theoretical price of the period-product payment} \\ 
 & -\text{the average of the theoretical price of the period-product payment over paths} \\ 
\end{aligned}	
\end{equation}

The concept of the period-product has another use, as we cannot directly observe the instantaneous value for the type 1 performance variables; we can only observe the discrete samples of the period-sum for each accounting period. All the model parameters relate to instantaneous variables, but these cannot be directly estimated. Let $k$ be the number of observable accounting periods. Given the instantaneous variable parameters (such as volatility or correlation, among others), we generate a set of period-sum samples 
$\mathbf{D}_{SUM}^{\left( h \right)}:=\left\{ P_{iSUM}^{\left( h \right)}\left( {{t}_{1}},{{t}_{2}} \right),\cdots ,P_{iSUM}^{\left( h \right)}\left( {{t}_{k}},{{t}_{k+1}} \right) \right\}$ for $i\in {{N}_{1}}$ 
 of the $h$-th trial by the Monte Carlo method and obtain one statistic $f\left( \mathbf{D}_{SUM}^{\left( h \right)} \right)$of the $h$th trial. We repeat this procedure $H$ times and obtain the average ${{e}_{SUM}}:={\sum\nolimits_{h}{f\left( \mathbf{D}_{SUM}^{\left( h \right)} \right)}}/{H}\;$. Similarly, we obtain ${{e}_{PROD}}:={\sum\nolimits_{h}{f\left( \mathbf{D}_{PROD}^{\left( h \right)} \right)}}/{H}\;$. At the same time, we know the analytical value corresponding to ${{e}_{PROD}}$ through the above analytical distribution formula, which is denoted by ${{a}_{PROD}}$. We can use the parameter of the period-product as a control variable to improve the accuracy in estimating the period-sum parameter:
\begin{equation}
	\text{The estimation of the period-sum parameter}={{e}_{SUM}}-{{e}_{PROD}}+{{a}_{PROD}}.	
\end{equation}

{\flushleft{Thus, we determine the instantaneous variable parameters so the left side of (4.19) agrees with the real observation values.}}

Many points of caution are involved in estimating parameters. Fix arbitrarily $1\le i\le m$ and $1\le j\le n$; ${{S}_{i}}$ and ${{P}_{j}}$ are denoted as $S$ and $P$ for simplicity.
\begin{arabiclist}
\item When we estimate the correlation between the stock price and performance variable, as the observable performance variable is the period-sum of each accounting period $\left[ {{t}_{i}},{{t}_{i+1}} \right]\ \left( i=1,\cdots ,k \right)$, we must use the period-sum of the stock price in the corresponding period 
\begin{equation}
	{{S}_{SUM}}\left( {{t}_{i}},{{t}_{i+1}} \right):=\int_{{{t}_{i}}}^{{{t}_{i+1}}}{{{S}_{u}}du}\ \left( i=1,\cdots ,k \right).	
\end{equation}

{\flushleft{When the number of issued shares changes during these periods, the period-sums of the market capitalization are appropriate variables.}}
\begin{figure}[pt]
\centerline{\includegraphics[width=3.0in]{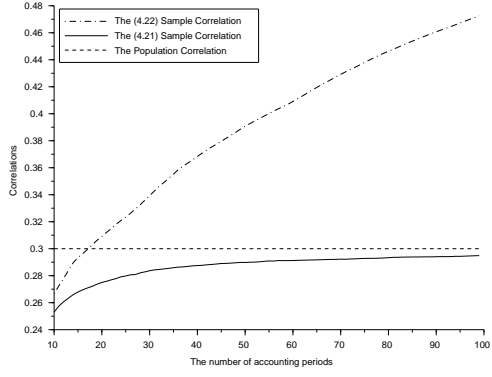}}
\vspace*{8pt}
\caption{The number of accounting periods and correlations.}
\end{figure}
\item It is correct to compute the sample correlation of the ratio's logarithm of successive two period-sums
\begin{equation}
	\left( \log \frac{{{S}_{SUM}}\left( {{t}_{i+1}},{{t}_{i+2}} \right)}{{{S}_{SUM}}\left( {{t}_{i}},{{t}_{i+1}} \right)},\log \frac{{{P}_{SUM}}\left( {{t}_{i+1}},{{t}_{i+2}} \right)}{{{P}_{SUM}}\left( {{t}_{i}},{{t}_{i+1}} \right)} \right)\ \left( i=1,\cdots ,k-1 \right),	
\end{equation}

{\flushleft{not one of the period-sum itself}}
\begin{equation}
	\left( {{S}_{SUM}}\left( {{t}_{i}},{{t}_{i+1}} \right),{{P}_{SUM}}\left( {{t}_{i}},{{t}_{i+1}} \right) \right)\ \left( i=1,\cdots ,k-1 \right).	
\end{equation}

{\flushleft{The latter correlation converges to one as the number of the accounting periods increases because the drift term eventually dominates. Fig. 1 illustrates the changes in the two sample correlations when the number of accounting periods increases, calculated by a Monte Carlo Simulation with 20,000 paths. We assume the population correlation is 0.3. Although the sample correlation of the ratio's logarithm converges to a population correlation of 0.3 as the accounting period number increases, while the sample correlation from (4.20) increases separate from the population correlation.}}

\item The sample correlation is not an unbiased estimator. The approximate value of the unbiased estimator of the population correlation is the following equation (Olkin \& Pratt 1958):
\begin{equation}
	\hat{\rho }\left( 1+\frac{\left( 1-{{{\hat{\rho }}}^{2}} \right)}{2n-6} \right),	
\end{equation}

{\flushleft{where $\hat{\rho }$ is a sample correlation and $n$ is the number of samples, which is $k-1$ in our case.}}
\item When using a stock index as the optimal portfolio, a single-regression beta of the optimal portfolio is necessary; we can use the following formula:
\begin{equation}
\label{ZEqnNum151138}
	\begin{aligned}
  & \text{the estimation of the population correlation } \\ 
  & \text{between the stock index and the performance variable} \\ 
 & \times \frac{\text{ the estimation of volatility of the performance variable}}{\text{ the estimation of volatility of the stock index}}. \\ 
\end{aligned}	
\end{equation}

{\flushleft{We must estimate the population correlation between the stock index and the performance variable carefully as the estimation of that between the stock price and the performance variable.}}
\item We use the unbiased variance of samples of the ratio's logarithm of successive two period-sums as the population variance of the performance variable:
\begin{equation}
	\left( \log \frac{{{P}_{SUM}}\left( {{t}_{i+1}},{{t}_{i+2}} \right)}{{{P}_{SUM}}\left( {{t}_{i}},{{t}_{i+1}} \right)} \right)\ \ \left( i=1\cdots ,k-1 \right).	
\end{equation}

{\flushleft{We use the square root of the obtained population variance as an estimate of the performance variable's volatility.}}
\item We estimate the drift coefficient ${{c}_{i}}$ without beta-adjusting of the type 1 performance variable, by using the actual result and the market-consensus forecast of the period-sum in the following equation:
\begin{equation}
	\begin{aligned}
  & \text{the market-consensus forecast of the period-sum of year }Y \\ 
 & \text{=the realized  period-sum of year }{{Y}_{0}}\times \exp \left( {{c}_{i}}\left( Y-{{Y}_{0}} \right) \right). \\ 
\end{aligned}	
\end{equation}

{\flushleft{If there are multiple periods of performance goals, then the term-structure of the drift coefficients can be used.}}
\end{arabiclist}

\section{Implications for Accounting Standards }

The current accounting standards for share-based payments are as follows (Stock-Based Payments, Statement of Financial Accounting Standards No. 123 (R)):
\footnote{The IFRS 2 Share-Based Payment is essentially the same.}
\begin{arabiclist}
\item Measure the fair value of the awards on a grant date as with no performance conditions. Using an equation similar to (4.10), we set the below $A$ as the fair value:
\begin{equation}
	A:={{E}^{Q}}\left[ S_{0T}^{-1}C \right]	
\end{equation}

\item The company judges whether it is probable to meet these conditions on its grant date.
\item When judged as probable, allocate the fair value proportionately over the remaining service period.
\item When judged as improbable, do not recognize compensation cost.
\item When the performance condition results are known, if the performance conditions are not achieved, already recognized compensation cost will be reversed; in the case of achievement, unrecognized compensation cost will be immediately recorded.
\end{arabiclist}

We compare the current and new standards by defining new standards using an accounting process that measures the awards' fair value with the theoretical price $p$ of (4.10), which is always recognized as compensation cost on the grant date. In calculating $p$, we use the drift of the original probability distribution, minus the beta multiplied by an optimal portfolio's excess expected returns. As the probability of goal achievement in the original distribution is less than one and an additional downward beta adjustment exists, $p$ would be smaller than $A$. Further, $p$ and $A$ only coincide when the original distribution's achievement probability is one and the beta is zero. In an extreme case, in which the original distribution's goal achievement is one or zero, no performance condition-based motivation exists; thus, we are interested in the range, such as from 30\% to 70\%, for example. In this case, $p$ can be considered considerably smaller than $A$. Further, compensation cost costs doed not change, as they are always recognized regardless of whether the performance conditions are satisfied.

For example, consider that compensation cost for $N$ years and the performance condition results will be known in the $N$th year. Compensation cost under the current standards can only change in the $N$th year. We focus on the difference between recognized compensation cost in the $N$th (${{c}_{N}}$) and $N-1$th years (${{c}_{N-1}}$). Further, ${{I}_{g}}$ denotes a variable that assumes a value of one when the goal is achieved and zero when it is not, and ${{I}_{p}}$ denotes a variable that assumes a value of one when the company judges on the grant date that goal achievement is probable, and zero when the company judges this as improbable. As total compensation cost to be recognized for $N$ years depends only on goal achievement, we have
\begin{equation}
	A{{I}_{g}}={{c}_{N}}+\left( N-1 \right){{c}_{N-1}}.
\end{equation}

{\leftline{At the same time, ${{c}_{N-1}}$ depends only on the company's judgment:}}
\begin{equation}
	{{c}_{N-1}}=\frac{A{{I}_{p}}}{N}.
\end{equation}

{\leftline{From (5.2) and (5.3), we obtain}}
\begin{equation}
	{{c}_{N}}-{{c}_{N-1}}=A{{I}_{g}}-\left( N-1 \right){{c}_{N-1}}-{{c}_{N-1}}=A{{I}_{g}}-N{{c}_{N-1}}=A\left( {{I}_{g}}-{{I}_{p}} \right).
\end{equation}

Therefore, when ${{I}_{p}}$ is 1, ${{c}_{N}}-{{c}_{N-1}}$ is $-A$ or 0, and when ${{I}_{p}}$ is 0, then ${{c}_{N}}-{{c}_{N-1}}$ is 0 or $A$. The volatility (standard deviation) of ${{c}_{N}}-{{c}_{N-1}}$ is 
\begin{equation}
	\sqrt{\alpha \left( 1-\alpha  \right)}A,
\end{equation}
{\flushleft{where $\alpha $ is the goal achievement probability.}}
We observe that the volatility of ${{c}_{N}}-{{c}_{N-1}}$ depends only on the size of $A$ and the goal achievement probability $\alpha $, and does not depend on the company's choice. 

The expected value of total compensation cost is $\alpha A$, which also depends only on $A$ and the probability of achieving the goal $\alpha $, and not on the company's choice.

The problems with the current standards are as follows:
\begin{arabiclist}
\item Volatile compensation cost: On the one hand, volatile compensation cost under the current standards is inevitable, as indicated by (5.5). On the other hand, compensation cost is constant under the new standards.
\item Over-recognition of compensation cost: In the favorable case in which goals are achieved, if the company can correctly predict this achievement and no volatility effect occurs on the income statement, the larger compensation cost is recognized under the current standards, rather than the new. If $p$ is 40\% of $A$, the current standards compensation cost is 2.5 times the new standards. Under the current standards, the only unfavorable case in which the goal is not achieved involves a cost lower than under the new standards.

\item Inconsistent accounting objectives: Consider two awards. Award X has a low probability of goal achievement and many exercisable shares. Award Y has a high probability of goal achievement and few exercisable shares. Suppose the service provider is indifferent to both awards. The purpose of measuring the award's fair value is to measure the fair value of the service, and to recognize it as compensation cost. This is in line with (ASC 718-10-10-1)
\footnote{http://guides.newman.baruch.cuny.edu/FASB\_Codification/citing}
, which notes, “The objective of accounting for transactions under share-based payment arrangements with employees is […] to recognize in the financial statements the employee services received [...] and the related cost to the entity as those services are consumed." The difference between the fair value of awards X and Y under the current standards is
\begin{equation}
	\left( {{q}_{X}}-{{q}_{Y}} \right)A
\end{equation}

{\flushleft{where ${{q}_{X\left( Y \right)}}$ is the number of exercisable shares for award X (or Y). Alternatively, the difference under the new standards is}}
\begin{equation}
	{{q}_{X}}{{p}_{X}}-{{q}_{Y}}{{p}_{Y}}=\left( {{q}_{X}}-{{q}_{Y}} \right){{p}_{X}}+{{q}_{Y}}\left( {{p}_{X}}-{{p}_{Y}} \right)
\end{equation}

{\flushleft{where ${{p}_{X\left( Y \right)}}$ is the theoretical price {\textit{per}} exercisable share of award X (or Y). The first term on the right side of (5.7) is smaller than (5.6) because ${{p}_{X}}$ is much smaller than $A$, and the second term is negative because ${{p}_{X}}$ is much smaller than ${{p}_{Y}}$. Thus, ${{q}_{X}}{{p}_{X}}-{{q}_{Y}}{{p}_{Y}}$ is much smaller than $\left( {{q}_{X}}-{{q}_{Y}} \right)A$. Ideally, the measured values for the same services will be the same regardless of the compensation scheme, and the new standards are clearly superior to the current standards.}}
\item Distorting a company's optimal award selection: A company has a bias to avoid high-risk compensation awards. If a company adopts award X and judges it as probable, the profit decreases by large compensation cost. If the company judges this as improbable, this conveys to the market that this goal will be difficult to achieve. Both outcomes are not preferable for the company. Therefore, the company has a bias to adopt the low-risk award Y. Such a dilemma is lessened in the new standards, which have a low possibility of distorting the company's decision.
\item The significant volatility of a difficult project's compensation cost: In cases that involve the development of new technologies or drugs, for example, the goal completely differs from such goals as the earnings per share growth rate, which the company can freely select as its goal and the probability of achievement from a continuum. As achieving the goal is difficult, it is necessary to increase the benefit obtained at the time of achievement. In the above example, the company has no choice other than to select the award X. Under the current standards, because $A$ is large, compensation cost is significantly volatile, as indicated by (5.5). This volatility decreases the income statement's reliability. Under the new standards, as recognized compensation cost is constant, the income statement retains its reliability.
\end{arabiclist}

When the compulsory use of the fair value method was mandated in 2004, the current standards were the only choice, as no model properly reflected performance conditions in its theoretical pricing. However, once a model appropriately reflects performance conditions, there is no reason to keep the current standards. We posit it is worthwhile to consider the new standards proposed herein as well as their improved version.

\section{Conclusions }
Although the SPPC has experienced prominent growth, theoretical pricing has not been sufficiently studied. We examined previous studies' results regarding the theoretical price of contingent claims in an incomplete market, and incorporated the concept of a marginal utility-based price. We then adopted an approach to restrict the stochastic processes to a certain class, which is nonrestrictive and hence aggregable in application. We demonstrated a need for consistent change in the probability distributions of stock prices as well as the performance variables that affect the payoff, then developed two models. The second model, which uses an optimal portfolio, is incredibly convenient and persuasive in its application, as it uses only a few parameters. We then provided a method to estimate these parameters and improve the estimation. Simultaneously, we demonstrated that the current accounting standards—specifically, the Share-Based Payment No. 123 (R)—have some defects, which our theoretical price model can greatly improve.

\appendix
\section{Karatzas-Shreve conditions}
We call the following conditions \textit{Karatzas-Shreve conditions} (Karatzas \& Shreve 1998, Assumptions 3.8.1, 3.8.2).
\begin{arabiclist}

\item The processes $r$ and $\left\| \text{\boldmath$\theta$} \right\|$ are H\"{o}lder continuous; namely, for some $K>0$ and $\rho \in \left( 0,1 \right)$ we have
\begin{equation}
	\left| r\left( {{t}_{1}} \right)-r\left( {{t}_{2}} \right) \right|\le K{{\left| {{t}_{1}}-{{t}_{2}} \right|}^{\rho }},{\text{ }}\left| \left\| \text{\boldmath$\theta$}\left( {{t}_{1}} \right) \right\|-\left\| \text{\boldmath$\theta$}\left( {{t}_{2}} \right) \right\| \right|\le K{{\left| {{t}_{1}}-{{t}_{2}} \right|}^{\rho }}	
\end{equation}

{\flushleft{for all ${{t}_{1}},{{t}_{2}}\in \left[ 0,T \right]$.}}

\item Some positive constants ${{k}_{1}},{{k}_{2}}$ exist, such that 
\begin{equation}
\label{ZEqnNum574647}
	{{k}_{1}}\le \left\| \text{\boldmath$\theta$}\left( t \right) \right\|\le {{k}_{2}},\quad \forall t\in \left[ 0,T \right].
\end{equation}
\end{arabiclist}

The utility function $U$ and the inverse function of the marginal utility $I:={{\left( {{U}'} \right)}^{-1}}$ satisfy: 
\begin{arabiclist}

\item (Polynomial growth of $I$) A constant $\gamma >0$ exists, such that
\begin{equation}
\label{ZEqnNum893452}
	I\left( y \right)\le \gamma +{{y}^{\gamma }};
\end{equation}

\item (polynomial growth of $U\left( I \right)$ A constant $\gamma >0$ exists, such that
\begin{equation}
\label{ZEqnNum296871}
	U\left( I\left( y \right) \right)\ge -\gamma -{{y}^{\gamma }},\quad \forall y\in \left( 0,\infty  \right).
\end{equation}
\end{arabiclist}

\section{An analytical distribution formula}
We assume ${{L}_{i}}=1\,{\text{ for }}\ i\in {{N}_{2}}$ to simplify the description. Let ${{N}_{0}}={{\left( i \right)}_{1\le i\le m}}$. We can illustrate that a $d$-dimensional random column vector
\begin{equation}
\label{ZEqnNum572183}
	\begin{aligned}
  & {{\left( {{\left( \log {{S}_{i}}\left( T \right) \right)}_{i\in {{N}_{0}}}},{{\left( \log {{P}_{iPROD}} \right)}_{i\in {{N}_{1}}}},{{\left( \log {{{P}_{i}}\left( {{t}_{i}}_{1} \right)}/{{{P}_{i}}\left( {{t}_{i}}_{0} \right)}\; \right)}_{i\in {{N}_{2}}}},{{\left( \log {{P}_{i}}\left( {{T}_{v}} \right) \right)}_{i\in {{N}_{3}}}} \right)}^{\top }} \\ 
 & =:\left( {{x}_{hi}} \right)_{0\le h\le 3,i\in {{N}_{h}}}^{\top }=:\mathbf{x} \\ 
\end{aligned}	
\end{equation}

\leftline{has a $d$-dimensional normal distribution.}

We use the moment-generating function of $\mathbf{x}$ to prove this. As preparation, it is necessary to analyze each random variable.

For $1\le i\le m$ and $1\le j\le d$, let ${{d}_{i}}$ be a dividend yield of the $i$th stock, ${{\sigma }_{i}}_{j}$ be the $\left( i,j \right)$th element of $\left[ \begin{matrix}
   {{\mathbf{\Sigma }}_{1}} & 0  \\
\end{matrix} \right]$, and ${{\hat{w}}_{j}}$ be the $j$th element of $\mathbf{\hat{w}}$.

For $i\in {{N}_{0}}$, we have
\begin{equation}
\label{ZEqnNum803722}
	\log {{S}_{i}}\left( T \right)={{m}_{0}}_{i}+\sum\limits_{j=1}^{d}{\int_{0}^{T}{\tilde{T }{{}_{0}}{{_{i}}_{j}}\left( t \right)d{{{\hat{w}}}_{j}}\left( t \right)}}	
\end{equation}

\leftline{where}
\begin{equation}
	{{m}_{0}}_{i}:=\log {{S}_{i}}\left( 0 \right)+\int_{0}^{T}{\left( r\left( t \right)-{{d}_{i}}\left( t \right)-\frac{1}{2}\sum\limits_{j=1}^{d}{\sigma _{ij}^{2}\left( t \right)} \right)dt}	
\end{equation}

\leftline{and}
\begin{equation}
	\tilde{T }{{}_{0}}{{_{i}}_{j}}\left( t \right):={{\sigma }_{i}}_{j}\left( t \right).	
\end{equation}

For $1\le i\le n$ and $1\le j\le d$, let ${{\mu }_{i}}$ be the $i$th element of $\text{\boldmath$\mu $}$ in (4.1), ${{T }_{i}}_{j}$ be the $\left( i,j \right)$th element of $\mathbf{T }$, and $\sigma _{i}^{2}$ be $\sum\limits_{j=1}^{d}{T _{ij}^{2}}$.

From (3.1), we have
\begin{equation}
	\log {{P}_{i}}\left( t \right)=\log {{P}_{i}}\left( 0 \right)+\int_{0}^{t}{\left( {{\mu }_{i}}\left( u \right)-\frac{1}{2}\sigma _{i}^{2}\left( u \right) \right)du}+\sum\limits_{j=1}^{d}{\int_{0}^{t}{{{T }_{i}}_{j}\left( u \right)d\hat{w}{{}_{j}}\left( u \right)}}.	
\end{equation}

For $i\in {{N}_{1}}$, substituting (B.5) into the right side of the logarithm of (4.11) yields
\begin{equation}
	\begin{aligned}
  & \log {{P}_{iPROD}}=\log \left( {{b}_{i}}-{{a}_{i}} \right){{P}_{i}}\left( 0 \right)+\frac{1}{{{b}_{i}}-{{a}_{i}}}\int_{{{a}_{i}}}^{{{b}_{i}}}{\int_{0}^{t}{\left( {{\mu }_{i}}\left( u \right)-\frac{1}{2}\sigma _{i}^{2}\left( u \right) \right)du}dt} \\ 
 & +\frac{1}{{{b}_{i}}-{{a}_{i}}}\int_{{{a}_{i}}}^{{{b}_{i}}}{\left( \sum\limits_{j=1}^{d}{\int_{0}^{t}{{{T }_{i}}_{j}\left( u \right)d\hat{w}{{}_{j}}\left( u \right)}} \right)dt}.  
\end{aligned}	
\end{equation}

{\flushleft{We then apply a generalized form of Fubini's theorem for stochastic integrals (Heath \& Morton 1992) to the random terms in (B.6) to have}}
\begin{equation}
\label{ZEqnNum954581}
\begin{aligned}
  & \int_{{{a}_{i}}}^{{{b}_{i}}}{\left( \int_{0}^{t}{{{T }_{i}}_{j}\left( u \right)d\hat{w}{{}_{j}}\left( u \right)} \right)dt}=\int_{{{a}_{i}}}^{{{b}_{i}}}{\left( {{b}_{i}}-u \right){{T }_{i}}_{j}\left( u \right)d\hat{w}{{}_{j}}\left( u \right)}+\left( {{b}_{i}}-{{a}_{i}} \right)\int_{0}^{{{a}_{i}}}{{{T }_{i}}_{j}\left( u \right)d\hat{w}{{}_{j}}\left( u \right)} \\ 
 & =\int_{0}^{T}{\left( \left( {{b}_{i}}-u \right){{T }_{i}}_{j}\left( u \right){{1}_{\left\{ {{a}_{i}}\le u\le {{b}_{i}} \right\}}}+\left( {{b}_{i}}-{{a}_{i}} \right){{T }_{i}}_{j}\left( u \right){{1}_{\left\{ 0\le u\le {{a}_{i}} \right\}}} \right)d\hat{w}{{}_{j}}\left( u \right)}.  
\end{aligned}	
\end{equation}

\leftline{Substituting (B.7) into (B.6), we have}
\begin{equation}
\label{ZEqnNum705888}
	\log {{P}_{iPROD}}={{m}_{1}}_{i}+\sum\limits_{j=1}^{d}{\int_{0}^{T}{\tilde{T }{{}_{1}}{{_{i}}_{j}}\left( t \right)d\hat{w}{{}_{j}}\left( t \right)}},	
\end{equation}

\leftline{where}
\begin{equation}
\label{ZEqnNum436749}
	{{m}_{1}}_{i}:=\log \left( {{b}_{i}}-{{a}_{i}} \right){{P}_{i}}\left( 0 \right)+\frac{1}{{{b}_{i}}-{{a}_{i}}}\int_{{{a}_{i}}}^{{{b}_{i}}}{\left( \int_{0}^{t}{\left( {{\mu }_{i}}\left( u \right)-\frac{1}{2}\sigma _{i}^{2}\left( u \right) \right)du} \right)dt}	
\end{equation}

\leftline{and}
\begin{equation}
	\tilde{T }{{}_{1}}{{_{i}}_{j}}\left( t \right):=\frac{\left( \left( {{b}_{i}}-t \right){{T }_{i}}_{j}\left( t \right){{1}_{\left\{ {{a}_{i}}\le t\le {{b}_{i}} \right\}}}+\left( {{b}_{i}}-{{a}_{i}} \right){{T }_{i}}_{j}\left( t \right){{1}_{\left\{ 0\le t\le {{a}_{i}} \right\}}} \right)}{{{b}_{i}}-{{a}_{i}}}.	
\end{equation}

For $i\in {{N}_{2}}$, from (B.5), we have 
\begin{equation}
	\log \frac{{{P}_{i}}\left( {{t}_{i_{1}}} \right)}{{{P}_{i}}\left( {{t}_{i_{0}}} \right)}={{m}_{2i}}+\sum\limits_{j=1}^{d}{\int_{0}^{T}{\tilde{T }{{}_{2}}{{_{i}}_{j}}\left( t \right)d\hat{w}{{}_{j}}\left( t \right)}},	
\end{equation}

\leftline{where }
\begin{equation}
	{{m}_{2i}}:=\int_{{{t}_{i_{0}}}}^{{{t}_{i_{1}}}}{\left( {{\mu }_{i}}\left( t \right)-\frac{1}{2}\sigma _{i}^{2}\left( t \right) \right)dt}	
\end{equation}

\leftline{and}
\begin{equation}
	\tilde{T }{{}_{2}}{{_{i}}_{j}}\left( t \right):={{T }_{i}}_{j}\left( t \right){{1}_{\left\{ {{t}_{i}}_{0}\le t\le {{t}_{i}}_{1} \right\}}}.	
\end{equation}

For $i\in {{N}_{3}}$, from (B.5), we have
\begin{equation}
	\log {{P}_{i}}\left( {{T}_{v}} \right)={{m}_{3}}_{i}+\sum\limits_{j=1}^{d}{\int_{0}^{T}{\tilde{T }{{}_{3}}{{_{i}}_{j}}\left( t \right)d\hat{w}{{}_{j}}\left( t \right)}},	
\end{equation}

\leftline{where}
\begin{equation}
	{{m}_{3i}}:=\log {{P}_{i}}\left( 0 \right)+\int_{0}^{{{T}_{v}}}{\left( {{\mu }_{i}}\left( t \right)-\frac{1}{2}\sigma _{i}^{2}\left( t \right) \right)dt},	
\end{equation}

\leftline{and}
\begin{equation}
	\tilde{T }{{}_{3}}{{_{i}}_{j}}\left( t \right):={{T }_{3}}{{_{i}}_{j}}\left( t \right){{1}_{\left\{ 0\le t\le {{T}_{v}} \right\}}}.	
\end{equation}

Let $\text{\boldmath$\theta$}:={{\left( {{\left( {{\theta }_{0}}_{i} \right)}_{1\le i\le m}},{{\left( {{\theta }_{1}}_{i} \right)}_{i\in {{N}_{1}}}},{{\left( {{\theta }_{2}}_{i} \right)}_{i\in {{N}_{2}}}},{{\left( {{\theta }_{3}}_{i} \right)}_{i\in {{N}_{3}}}} \right)}^{\top }}$ be a column vector, which is the coefficient in the moment-generating function of $\mathbf{x}$. The moment-generating function $M\left( \text{\boldmath$\theta$} \right)$ of $\mathbf{x}$ is 
\begin{equation}
\label{ZEqnNum212455}
	\begin{aligned}
  M\left( \text{\boldmath$\theta$} \right)=& {{E}^{Q}}\left[ \exp \left( {{\text{\boldmath$\theta$}}^{\top }}\mathbf{x} \right) \right] \\ 
 =& \exp \left( \sum\limits_{h,i\in {{N}_{h}}}{{{\theta }_{h}}_{i}{{m}_{hi}}} \right)\prod\limits_{1\le j\le d}{{{E}^{Q}}\left[ \exp \left\{ \int_{0}^{T}{\sum\limits_{h,i\in {{N}_{h}}}{{{\theta }_{h}}_{i}\tilde{T }{{}_{hi}}_{j}\left( t \right)}d{{{\hat{w}}}_{j}}\left( t \right)} \right\} \right]} \\ 
 =& \exp \left( \sum\limits_{h,i\in {{N}_{h}}}{{{\theta }_{h}}_{i}{{m}_{hi}}} \right)\prod\limits_{1\le j\le d}{\exp \left\{ \frac{1}{2}\int_{0}^{T}{{{\left( \sum\limits_{h,i\in {{N}_{h}}}{{{\theta }_{h}}_{i}\tilde{T }{{}_{hi}}_{j}}\left( t \right) \right)}^{2}}dt} \right\}} \\ 
 =& \exp \left( \sum\limits_{h,i\in {{N}_{h}}}{{{\theta }_{h}}_{i}{{m}_{hi}}} \right)\\
 &    \times \prod\limits_{1\le j\le d}{\exp \left\{ \frac{1}{2}\int_{0}^{T}{\left( \sum\limits_{h,i\in {{N}_{h}}}{\sum\limits_{{h}',{i}'\in {{N}_{{{h}'}}}}{{{\theta }_{h}}_{i}{{\theta }_{{{h}'}}}_{{{i}'}}\tilde{T }{{}_{hi}}_{j}\left( t \right)\tilde{T }{{}_{{h}'{i}'}}_{j}\left( t \right)}} \right)dt} \right\}} \\ 
 =& \exp \left( \sum\limits_{h,i\in {{N}_{h}}}{{{\theta }_{h}}_{i}{{m}_{hi}}} \right)\\
 &  \times \exp \left\{ \frac{1}{2}\sum\limits_{j=1}^{d}{\left( \int_{0}^{T}{\left( \sum\limits_{h,i\in {{N}_{h}}}{\sum\limits_{{h}',{i}'\in {{N}_{{{h}'}}}}{{{\theta }_{h}}_{i}{{\theta }_{{{h}'}}}_{{{i}'}}\tilde{T }{{}_{hi}}_{j}\left( t \right)\tilde{T }{{}_{{h}'{i}'}}_{j}\left( t \right)}} \right)dt} \right)} \right\} \\ 
 =& \exp \left( \sum\limits_{h,i\in {{N}_{h}}}{{{\theta }_{h}}_{i}{{m}_{hi}}} \right)\\
 &  \times \exp \left\{ \frac{1}{2}\sum\limits_{h,i\in {{N}_{h}}}{\sum\limits_{{h}',{i}'\in {{N}_{{{h}'}}}}{{{\theta }_{h}}_{i}{{\theta }_{{{h}'}}}_{{{i}'}}\int_{0}^{T}{\left( \sum\limits_{j=1}^{d}{\tilde{T }{{}_{hi}}_{j}\left( t \right)\tilde{T }{{}_{{h}'{i}'}}_{j}\left( t \right)} \right)dt}}} \right\}.  
\end{aligned}
\end{equation}

\leftline{We used the moment-generating function of a normal random variable for the third}
\leftline{equality of (B.17) (Shreve 2004, 4.4.30).}

The last side in (B.17) indicates $\mathbf{x}$ has a $d$-dimensional normal distribution, which has means of ${{\left( {{m}_{hi}} \right)}_{0\le h\le 3,i\in {{N}_{h}}}}$ and covariances between for $i\in {{N}_{h}}$ $\left( 0\le h\le 3 \right)$  and for ${i}'\in {{N}_{{{h}'}}}$ $\left( 0\le {h}'\le 3 \right)$ of $\int_{0}^{T}{\sum\nolimits_{j}{\tilde{T }{{}_{hi}}_{j}\left( t \right)\tilde{T }{{}_{{h}'{i}'}}_{j}\left( t \right)}dt}$. We can express every covariance using the instantaneous covariances. For example, we have
\begin{equation}
\begin{aligned}
&\int_{0}^{T}{\sum\nolimits_{j}{\tilde{T }{{}_{0i}}_{j}\left( t \right)\tilde{T }{{}_{1k}}_{j}\left( t \right)}dt}
\\
&=\frac{\int_{{{a}_{k}}}^{{{b}_{k}}}{\left( {{b}_{i}}-t \right)\sum\nolimits_{j}{{{\sigma }_{i}}_{j}\left( t \right){{T }_{k}}_{j}\left( t \right)}dt}+\left( {{b}_{k}}-{{a}_{k}} \right)\int_{0}^{{{a}_{k}}}{\sum\nolimits_{j}{{{\sigma }_{i}}_{j}\left( t \right){{T }_{k}}_{j}\left( t \right)}dt}}{{{b}_{k}}-{{a}_{k}}},
\end{aligned}
\end{equation}

\leftline{where $\sum\nolimits_{j}{{{\sigma }_{i}}_{j}\left( t \right){{T }_{k}}_{j}\left( t \right)}$ is an instantaneous covariance between ${d{{S}_{i}}}/{{{S}_{i}}}\;$and ${d{{P}_{k}}}/{{{P}_{k}}}\;$.}


\begin{thebibliography}{100}
\bibitem[Bingham & Kiesel (2013)]{dummy9,} N. H. Bingham \& R. Kiesel (2013) {\it Risk-Neutral Valuation: Pricing and Hedging of Financial Derivatives}. London: Springer Science \& Business Media.
\bibitem[Carhart (1997)]{dummy10,} M. M. Carhart (1997) On persistence in mutual fund performance, {\it The Journal of Finance} {\bf52} (1), 57-62.
\bibitem[Davis (1997)]{dummy11,} M. Davis (1997) Option pricing in incomplete markets, {\it Mathematics of Derivative Securities} {\bf15}, 216-226.
\bibitem[Delbaen & Schachermayer (1997)]{dummy12,} F. Delbaen \& W. Schachermayer (1997) The Banach space of workable contingent claims in arbitrage theory, {\it Probability and Statistics} {\bf33} (1), 113-144.
\bibitem[Delbaen & Schachermayer (1994)]{dummy13,} F. Delbaen \& W. Schachermayer (1994) A general version of the fundamental theorem of asset pricing, {\it Mathmatische Annalen} {\bf300} (1), 463-520.
\bibitem[Fama & French (2010)]{dummy14,} E. F. Fama \& K. French (2010) Luck versus skill in the cross-section of mutual fund returns, {\it The Journal of Finance} {\bf65} (5), 1915-1947.
\bibitem[Frittelli (2000)]{dummy15,} M. Frittelli (2000) The minimal entropy martingale measure and the valuation problem in incomplete markets, {\it Mathematical Finance} {\bf10} (1), 39-52.
\bibitem[Heath & Morton (1992)]{dummy16,} D. R. Heath \& A. Morton (1992) Bond pricing and the term structure of interest rates: a new methodology for contingent claims valuation, {\it Econometrica} {\bf60} (1), 77-105.
\bibitem[Henderson (2002)]{dummy17,} V. Henderson (2002) Valuation of claims on nontraded assets using utility maximization, {\it Mathematical Finance} {\bf12} (4), 351-373.
\bibitem[Hugonnier \textit{et al}. (2005)]{dummy18,} J. Hugonnier, D. Kramkov \& W. Schachermayer (2005) On utility-based pricing of contingent claims in incomplete markets, {\it Mathematical Finance} {\bf15} (2), 203-212.
\bibitem[Karatzas & Shreve (2012)]{dummy19,} I. Karatzas \& S. E. Shreve (2012) {\it Brownian Motion and Stochastic Calculus}. Dordrecht: Springer Science \& Business Media.
\bibitem[Karatzas & Shreve (1998)]{dummy20,} I. Karatzas \& S. E. Shreve (1998) {\it Methods of Mathematical Finance}. New York: Springer Verlag.
\bibitem[Karatzas & Shreve (1991)]{dummy21,} I. Karatzas, J. P. Lehoczky, S. E. Shreve \& G. L. Xu (1991) Martingale and duality methods for utility maximization in an incomplete market, {\it SIAM Journal on Control and Optimization} {\bf29} (3), 702-730.
\bibitem[Kramkov & Schachermayer (1999)]{dummy23,} D. Kramkov \& W. Schachermayer (1999) The asymptotic elasticity of utility functions and optimal investment in incomplete markets, {\it The Annals of Applied Probability} {\bf9} (3), 904-950.
\bibitem[Olkin & Pratt (1958)]{dummy24,} I. Olkin \& J. W. Pratt (1958) Unbiased estimation of certain correlation coefficients, {\it The Annals of Mathematical Statistics} {\bf29} (1), 201-211.
\bibitem[Schweizer (1999)]{dummy25,} M. Schweizer (1999) A guided tour through quadratic hedging approaches, {\it Discussion Papers, Interdisciplinary Research Project 373: Quantification and Simulation of Economic Processes.} {\bf1999} (96).
\bibitem[Shreve (2004)]{dummy26,} S. E. Shreve (2004) {\it Stochastic Calculus for Finance II: Continuous-Time Models}. Dordrecht: Springer Science \& Business Media.
\bibitem[Williams (1991)]{dummy27,} D. Williams (1991) {\it Probability with Martingales}. Cambridge: Cambridge University Press.
\end{thebibliography}
\end{document}